\documentclass[a4paper,UKenglish,cleveref, autoref, thm-restate]{lipics-v2021}
\pdfoutput=1

\title{Finite-memory Strategies for Almost-sure Energy-MeanPayoff Objectives in MDPs}
\titlerunning{Finite-memory Strategies for Almost-sure Energy-MeanPayoff in MDPs}

\author{Mohan Dantam}{School of Informatics, University of Edinburgh, UK}{}{}{}
\author{Richard Mayr}{School of Informatics, University of Edinburgh, UK}{}{}{}

\authorrunning{M.~Dantam and R.~Mayr}
\Copyright{Mohan Dantam and Richard Mayr}


\usepackage[colorinlistoftodos,prependcaption,textsize=tiny,obeyDraft]{todonotes}

\usepackage{amsmath,amssymb,amsthm}
\usepackage{mathtools}
\usepackage{xcolor,pgf,tikz}
\usepackage{comment}
\usepackage{hyperref}
\hypersetup{colorlinks,linkcolor=blue,citecolor=blue,urlcolor=blue}
\usepackage{cleveref}
\usepackage{array}
\usepackage{makecell}
\usepackage[ruled,vlined]{algorithm2e}
\usepackage{lmodern}
\usetikzlibrary{arrows.meta,automata,shapes,decorations,positioning,trees,calc}

\newcommand{\ie}{\emph{i.e.},\ }
\DeclareMathOperator{\parity}{\mathtt{Parity}}

\newcommand{\LimInf}{\operatorname{\mathtt{LimInf}}}
\newcommand{\LimSup}{\operatorname{\mathtt{LimSup}}}
\newcommand{\MP}{\operatorname{\mathtt{MP}}}
\newcommand{\energymove}[1]{\stackrel{#1}{\movesto}}
\newcommand{\en}{\mathtt{EN}}
\newcommand{\infix}{\mathtt{Infix}}
\newcommand{\Gain}{\mathtt{Gain}}
\newcommand{\Bailout}{\mathtt{Bailout}}
\newcommand{\tendsto}{\rightarrow}

\newcommand{\Q}{\mathbb{Q}}

\DeclarePairedDelimiter{\set}{\{}{\}}
\DeclarePairedDelimiter{\ceil}{\lceil}{\rceil}
\DeclarePairedDelimiter{\floor}{\lfloor}{\rfloor}
\newtheorem{fact}{Fact}

\newcommand{\+}[1]{\mathbb{#1}}
\newcommand{\?}[1]{\mathcal{#1}}

\newcommand{\N}{\+{N}}

\newcommand{\R}{\+{R}}
\newcommand{\x}{\times}
\newcommand{\Ocompl}{\?{O}}

\usepackage{wasysym} %
\usepackage{xparse}
\usepackage{xifthen}
\usepackage{mleftright}
\newcommand{\rsymbol}{\ocircle}
\newcommand{\zsymbol}{\Box}

\newcommand{\zstates}{\states_\zsymbol}
\newcommand{\rstates}{\states_\rsymbol}

\newcommand{\reachset}{T}

\newcommand{\abs}[1]{\lvert#1\rvert}

\newcommand{\size}[1]{\abs{#1}}

\newcommand{\eqby}[2][=]{\stackrel{\text{{\tiny{#2}}}}{#1}}
\newcommand{\eqdef}{\eqby{def}}

\usepackage{bm}
\renewcommand{\vec}[1]{\bm{{#1}}} %%% changed vector notation from boldsymbol %%%

\newcommand{\eps}{\varepsilon}

\newcommand{\NP}{\ensuremath{\mathsf{NP}}}
\newcommand{\coNP}{\ensuremath{\mathsf{coNP}}}

\newcommand{\problemx}[3]{
\par\noindent\underline{\sc#1}\par\nobreak\vskip.2\baselineskip
\begingroup\clubpenalty10000\widowpenalty10000
\setbox0\hbox{\bf INPUT:\ }\setbox1\hbox{\bf QUESTION:\ }
\dimen0=\wd0\ifnum\wd1>\dimen0\dimen0=\wd1\fi
\vskip-\parskip\noindent
\hbox to\dimen0{\box0\hfil}\hangindent\dimen0\hangafter1\ignorespaces#2\par
\vskip-\parskip\noindent
\hbox to\dimen0{\box1\hfil}\hangindent\dimen0\hangafter1\ignorespaces#3\par
\endgroup}

\NewDocumentCommand{\Prob}{O{} O{} m}{\ifthenelse{\isempty{#3}}{\mathcal{P}^{#1}_{#2}}{\mathcal{P}^{#1}_{#2}\lrc{#3}}}
\NewDocumentCommand{\expectation}{O{} O{} m}{{\mathcal{E}^{#1}_{#2}\lrc{#3}}}
\NewDocumentCommand{\valueof}{O{} O{} m}{{\mathtt{val}^{#1}_{#2}\lrc{#3}}}
\NewDocumentCommand{\limval}{O{} O{} m}{{\mathtt{Lval}^{#1}_{#2}\lrc{#3}}}
\NewDocumentCommand{\ST}{O{} O{}}{\ifthenelse{\isempty{#2}}{\ifthenelse{\isempty{#1}}{\mathtt{ST}}{\mathtt{ST}\lrc{#1}}}{\mathtt{ST}\lrc{#1,#2}}}
\NewDocumentCommand{\winset}{O{} O{}}{\ifthenelse{\isempty{#2}}{\ifthenelse{\isempty{#1}}{\mathtt{Win}}{\mathtt{Win}_{#1}}}{\ifthenelse{\isempty{#1}}{\mathtt{Win}\lrc{#2}}{\mathtt{Win}_{#1}\lrc{#2}}}}

\newcommand{\dist}{\mathcal{D}}
\newcommand{\supp}{{\sf supp}}

\newcommand{\always}{{\sf G}}
\newcommand{\eventually}{{\sf F}}
\renewcommand{\next}{{\sf X}}
\newcommand{\obj}{\mathtt{O}}

\newcommand{\hide}[1]{}

\newcommand{\lrc}[1]{\mleft(#1\mright)}
\newcommand{\lrd}[1]{\{#1\}}

\newcommand{\ignore}[1]{}

\newcommand{\nat}{\mathbb N}

\newcommand{\setcomp}[2]{\lrd{{#1} \mid {#2}}}
\newcommand{\given}{{\,\mid\,}}
\newcommand{\tuple}[1]{\lrc{#1}}

\newcommand{\mdp}{{\mathcal M}}
\newcommand{\mc}{{\mathcal A}}

\newcommand{\mdptuple}{\tuple{\states,\zstates,\rstates,\transition,\probp,\vec{r}}}

\newcommand{\states}{S}

\newcommand{\s}{s}

\newcommand{\transition}{{E}}

\newcommand{\movesto}{{\longrightarrow}}

\newcommand{\probp}{P}

\newcommand{\complementof}[1]{\overline{#1}}

\newcommand{\play}{\rho}

\newcommand{\playsof}[1]{{\it Runs}\lrc{#1}}
\newcommand{\partialplay}{\rho}

\newcommand{\zstrat}{\sigma}

\newcommand{\xstrat}{\tau}
\newcommand{\optzstrat}{\zstrat^*}

\newcommand{\zallstrats}[1]{\zstratset^{{#1}}}

\newcommand{\zfinstrats}[1]{\zstratset^{{#1}}_{\finite}}
\newcommand{\finite}{f}
\newcommand{\zstratset}{\Sigma}

\newcommand{\pz}{\zsymbol}

\newcommand{\memory}{{\sf M}}
\newcommand{\updatefun}{upd}
\newcommand{\memconf}{{\sf m}}
\newcommand{\memconfa}{{\sf a}}
\newcommand{\memconfset}{{\sf M}}
\newcommand{\memsuc}{{\sf nxt}}

\newcommand{\memup}{{\sf \updatefun}}
\newcommand{\memstrattuple}{\tuple{\memory,\initmem,\memup,\memsuc}}
\newcommand{\initmem}{\memconf_0}

\newcommand{\om}{\omega}

\newcommand{\probm}{{\mathcal P}}

\newcommand{\coloring}{{\mathit{C}ol}}

\mathchardef\mhyphen="2D %

\newcommand{\Buchi}{{B\"uchi}}

\newcommand{\F}{{\mathcal F}}

\newcommand{\EN}[1]{\mathsf{EN}(#1)}
\newcommand{\ES}[1]{\mathsf{ST}(#1)}
\newcommand{\AS}{\mathtt{AS}}

\newcommand{\successors}[1]{\mathsf{Succ}({#1})}

\nolinenumbers

\begin{document}

\maketitle

\begin{abstract}
We consider finite-state Markov decision processes with 
the combined Energy-MeanPayoff objective.
The controller tries to avoid running out of energy while simultaneously attaining
a strictly positive mean payoff in a second dimension.

We show that \emph{finite memory} suffices for almost surely winning strategies for
the Energy-MeanPayoff objective.
This is in contrast to the closely related Energy-Parity objective, where
almost surely winning strategies require infinite memory in general.

We show that exponential memory is sufficient
(even for deterministic strategies)
and necessary
(even for randomized strategies)
for
almost surely winning Energy-MeanPayoff.
The upper bound holds even
if the strictly positive mean payoff part of the objective
is generalized to multidimensional strictly positive mean payoff.

Finally,
it is decidable in pseudo-polynomial time whether
an almost surely winning strategy exists.
\end{abstract}

\section{Introduction}\label{sec:intro}

\subparagraph{Background.}
Markov decision processes (MDPs) are a standard model for dynamic systems that
exhibit both stochastic and controlled behavior \cite{Puterman:book}.
MDPs play a prominent role in many domains, e.g., artificial intelligence and machine learning~\cite{sutton2018reinforcement,sigaud2013markov}, control theory~\cite{blondel2000survey,NIPS2004_2569}, operations research and finance~\cite{Sudderth:2020,Hill-Pestien:1987,bauerle2011finance,schal2002markov}, and formal verification~\cite{Flesch:JOTA2020,Sudderth:2020,FPS:2018,ModCheckHB18,ModCheckPrinciples08,chatterjee2012survey}.

An MDP is a directed graph where states are either controlled or random.
If the current state is controlled then the controller can choose a distribution over all possible successor states.
If the current state is random then the next state is chosen according to a
fixed probability distribution.
One assigns numeric rewards to transitions (and this can be generalized to
multidimensional rewards).
Moreover, priorities (aka colors), encoded by bounded non-negative numbers,
are assigned to states.
By fixing a strategy for the controller and an initial state, one obtains a probability space
of runs of the MDP. The goal of the controller is to optimize the expected value of
some objective function on the runs.

The \emph{strategy complexity} of a given objective is
the amount of memory (and randomization) needed for an optimal (resp.\ $\eps$-optimal)
strategy.
Common cases include memoryless strategies, finite-memory strategies, Markov
strategies (using a discrete clock, aka step counter), and general
infinite-memory strategies.

\subparagraph{Related work.}
The Parity, MeanPayoff and Energy objectives have been extensively studied in
the formal verification community.
A run satisfies the (min-even) {\em Parity objective} iff
the minimal priority that appears infinitely often in the run is even.
It subsumes all $\omega$-regular objectives, and in particular safety,
liveness, fairness, etc.
The {\em MeanPayoff objective} requires that the limit average reward per
transition along a run is positive (resp.\ non-negative in some settings).
MeanPayoff objectives go back to a 1957 paper by Gillette \cite{gillette1957stochastic} and have
been widely studied, due to their relevance for efficient control.
The {\em Energy objective} \cite{chakrabarti2003resource} requires that the
accumulated reward at any time in a run stays above some finite threshold
(typically $0$).
The intuition is that a controlled system has some finite initial energy level
that must never become depleted.

Combinations of these objectives have also been studied, where the runs need to
satisfy several of the above conditions simultaneously.

The existence of almost surely winning strategies for \emph{MeanPayoff-Parity}
in MDPs is decidable in polynomial time \cite{CD2011}.
These strategies require only finite memory for MeanPayoff $>0$ \cite{Gimbert2011ComputingOS},
but infinite memory for MeanPayoff $\ge 0$ \cite{CHJ2005}.

The existence of almost surely winning strategies for \emph{Energy-Parity}
in MDPs is decidable in $\NP \cap \coNP$ and in pseudo-polynomial time
\cite{MSTW2017}. (The $\NP \cap \coNP$ upper bound holds even for turn-based
stochastic games \cite{MSTW2021}.)
Almost surely winning strategies in MDPs require only finite memory in the special
case of Energy-B\"uchi \cite{CD2011}, but infinite memory for
Energy-co-B\"uchi and thus for Energy-Parity \cite{MSTW2017}.
However, $\eps$-optimal strategies for Energy-Parity require only finite
(at most doubly exponential) memory, and the value can be effectively approximated in doubly exponential
time (even for turn-based stochastic games) \cite{Dantam-Mayr:LIPIcs.MFCS.2023.38}.

The \emph{Energy-MeanPayoff} objective is similar to Energy-Parity, but
replaces the Parity part by a MeanPayoff objective for a second reward
dimension.
I.e., one considers an MDP with 2-dimensional transition rewards, where the
Energy condition applies to the first dimension and the MeanPayoff condition
applies to the second dimension. (It can be generalized to higher dimensions $d$,
where the MeanPayoff condition applies to all dimensions $2,3,\dots,d$.)
This might look like a direct generalization of the Energy-Parity objective, since
Parity games are reducible to MeanPayoff games \cite{Puri1995,Jurdzinski1998}.
However, this reduction does not work in the context of these combined
objectives when one considers stochastic systems like MDPs; see below.
Non-stochastic Energy-MeanPayoff games have been studied in \cite{bruyere_et_al:LIPIcs.CONCUR.2019.21}.

A sightly different objective has been studied in
\cite{Clemente-Raskin:2015} who
consider MDPs with $d$-dimensional rewards, where $d=d_1+d_2$.
The objective requires a strictly positive MeanPayoff
\emph{surely} in the first $d_1$ dimensions,
and \emph{almost surely} in the remaining $d_2$ dimensions.
This objective is strictly stronger than Energy-MeanPayoff.
E.g., a MeanPayoff of zero in the first dimension may or may not satisfy the Energy objective,
but it never satisfies the objective in \cite{Clemente-Raskin:2015}.

The objective studied in \cite{BKN2016} aims to maximize the expected
MeanPayoff (rather than the probability of it being strictly positive)
while satisfying the energy constraint.
However, unlike in our work, the reward function has a single dimension
(i.e., both criteria apply to the same value) and $\eps$-optimal strategies
can require infinite memory.

\subparagraph{Our contribution.}
We consider the Energy-MeanPayoff objective in MDPs with $d$-dimensional rewards.
The first dimension needs to satisfy the Energy condition (never drop below
$0$), while each other dimension needs to have a \emph{strictly} positive MeanPayoff.
We show that almost surely winning strategies for Energy-MeanPayoff require
only \emph{finite} memory.
\footnote{
Our results do \emph{not} carry over to Energy-MeanPayoff objectives
with \emph{non-strict} inequalities where
one just requires a MeanPayoff $\ge 0$ almost surely.
This needs infinite memory even for the case of $d=2$,
i.e., one energy-dimension and one MeanPayoff-dimension.
It suffices to modify the counterexample for Energy-co-B\"uchi from
\cite[Page 4]{MSTW2017} such that a visit to a state with unfavorable color
incurs a reward of $-1$ in the MeanPayoff-dimension.
}
This is in contrast to the Energy-Parity objective where almost surely winning
strategies require infinite memory in general \cite[Page 4]{MSTW2017}
(even for the simpler Energy-co-B\"uchi objectives).
This also shows that Energy-Parity is not reducible to Energy-MeanPayoff in
MDPs, unlike the reduction from Parity to MeanPayoff in
\cite{Puri1995,Jurdzinski1998}.

We show that almost surely winning strategies for Energy-MeanPayoff,
if they exist, can be chosen as deterministic strategies with
an exponential number of memory modes.
The crucial property is that it suffices to remember the stored energy only up to some
exponential upper bound.
A small counterexample shows the corresponding exponential lower bound.
Even for randomized strategies, an exponential number of memory modes is required,
and this holds even for the case of small transition rewards in $\{-1,0,+1\}$.

Although almost surely winning strategies are `exponentially large'
in this sense, %(even for small transition rewards),
their existence is still decidable in pseudo-polynomial time; cf.~\Cref{sec:conclusion}.
% Finally, the value of configurations (initial state and initial energy) can be
% effectively approximated along similar lines as
% for Energy-Parity in \cite{Dantam-Mayr:LIPIcs.MFCS.2023.38}.

\section{Preliminaries}\label{sec:prelim}

A \textit{probability distribution} over a countable set $S$ is a function
$f\!: S \to [0,1]$ with $\sum_{s \in S} f(s) = 1$.
$\supp(f) \eqdef$ \mbox{$\setcomp{s}{f(s)>0}$} denotes the support of
$f$ and $\dist(S)$ is the set of all probability distributions over $S$.
Given an alphabet $\Sigma$,
let $\Sigma^{\om}$ and $\Sigma^{*}$ ($\Sigma^+$) denote the set of infinite
and finite (non-empty) sequences over $\Sigma$, respectively.
Elements of $\Sigma^{\om}$ or $\Sigma^*$ are called words.

\subparagraph{MDPs and Markov chains.}
A \emph{Markov Decision Process} (MDP) is a controlled stochastic directed graph $\mdp \eqdef \mdptuple$
where the set of vertices $\states$ (also called states) is partitioned into the states $\zstates$ of
the player $\pz$ (\emph{Maximizer}),
and chance vertices (aka random states) $\rstates$.
Let $\transition \subseteq \states \x \states$ be the transition relation.
We write $\s \movesto \s'$ if $\tuple{\s,\s'} \in \transition$ and
assume that
$\successors{\s} \eqdef \{\s' \mid \s\transition{}\s'\} \neq \emptyset$
for every state $\s$.
The \emph{probability function}~$\probp$
assigns each random state $\s \in \rstates$ a distribution over
its successor states, i.e., $\probp(\s) \in \dist(\successors{\s})$.
We extend the domain of $\probp$ to
$\states^*\rstates$ by $\probp(\partialplay\s) \eqdef \probp(\s)$
for all $\partialplay\s \in \states^+\rstates$.
% Let $\obj$ be a generic objective and $\px \in \set{\pz,\po}$.
A \emph{Markov chain} is an MDP with only random states,
i.e., $\zstates = \emptyset$.
In this paper we consider finite-state MDPs, i.e., the set of states $S$ is finite.

\subparagraph{Strategies.}
A \textit{run} is an infinite sequence $\s_0\s_1 \ldots \in \states^{\omega}$
such that $\s_i \movesto \s_{i+1}$ for all $i \ge 0$.
A \textit{path} is a finite prefix of a run.
Let $\playsof{\mdp} \eqdef \set*{\play = \lrc{q_i}_{i \in \N} \, | q_i \movesto q_{i+1}}$
denote the set of all possible runs.
A strategy of the player $\pz$ is a function
$\zstrat\! : \states^* \zstates \to \dist(\states)$
that assigns to every path
$w\s \in \states^* \zstates$
a probability distribution over the successors of $\s$.
If these distributions are always Dirac then the strategy is called
\emph{deterministic} (aka pure), otherwise it is called \emph{randomized}
(aka mixed).
The set of all strategies of player $\pz$ in
$\mdp$ is denoted by $\zallstrats{\mdp}$.
A run/path $\s_0\s_1 \ldots$ is compatible with a strategy
$\zstrat$ if $\s_{i+1} \in \supp(\zstrat(\s_0 \ldots \s_i))$
whenever $\s_i \in \zstates$.
Finite-memory strategies are a subclass of
strategies using a finite set $\memconfset$ 
of memory modes.
A function $\memsuc : \memconfset \x \zstates \mapsto \dist(\states)$
chooses a (distribution over) successor states based on the current memory
mode and state
and
$\memup : \memconfset \x \transition \mapsto \dist(\memconfset)$
updates the memory mode upon observing a transition.
Let $\zstrat[\memconf]$ denote the finite-memory strategy $\zstrat$ starting
in memory mode $\memconf$.
The set of all finite-memory strategies in $\mdp$ is denoted by $\zfinstrats{\mdp}$.
Strategies with memory $|\memconfset|=1$ are called \emph{memoryless}.
Memoryless deterministic (resp.\ randomized) strategies are called MD (resp.\ MR).
By fixing some finite-memory strategy $\zstrat$ from some initial state in a finite-state MDP
$\mdp$, we obtain a finite-state Markov chain, denoted by $\mdp^\zstrat$. 

\subparagraph{Measure.} An MDP $\mdp$ with initial state $\s_0$ and strategy
$\zstrat$ yields a probability space
$(\s_0\states^{\om},\F_{\s_0}, \Prob[\mdp][\zstrat,\s_0]{})$
where $\F_{\s_0}$ is the $\sigma$-algebra generated by the cylinder sets
$\s_0\s_1\ldots\s_n\states^{\om}$ for $n \ge 0$.
The probability measure $\Prob[\mdp][\zstrat,\s_0]{}$ is
first defined on the cylinder sets.
For $\partialplay = \s_0\ldots\s_n$, let
$\Prob[\mdp][\zstrat,\s_0]{\partialplay} \eqdef 0$ if
$\partialplay$ is not compatible with $\zstrat$ and
otherwise 
$\Prob[\mdp][\zstrat,\s_0]{\partialplay\states^{\om}} \eqdef
\prod_{i=0}^{n-1} \xstrat(\s_0\ldots\s_i)(\s_{i+1}) $ where $\xstrat$
is $\zstrat$ or $\probp$ depending on
whether $\s_i \in \zstates$ or $\rstates$, respectively. 
If $\mdp$ is a Markov chain then there is only a single strategy, and we simply
write $\Prob[\mdp][\s_0]{}$.
By Carath\'eodory's extension
theorem~\cite{billingsley2008probability}, this defines a unique probability
measure on the $\sigma$-algebra.
Given some reward function $v: \s_0\states^{\om} \to \mathbb{R}$,
we write $\expectation[][]{.}$
for the expectation w.r.t.~$\probm$ and $v$.

\subparagraph{Objectives.}
General objectives are defined by real-valued measurable functions.
% and the associated payoff function for a pair of strategies
% $(\zstrat,\ostrat)$ and start state $\s_0$ is given by
% $\expectation[\game]{\zstrat,\ostrat,\s_0}(X)$\footnote{strictly speaking it
%  should be
%  $\expectation[\game]{\zstrat,\ostrat,\s_0}(X\restrict_{\s_0\states^{\om}})$
%  but we simply write as above for concise presentation} for player $\pz$ and
% $-\expectation[\game]{\zstrat,\ostrat,\s_0}(X)$ for player $\po$.
However, we mostly consider indicator functions of measurable sets.
Hence, our objectives can be described by measurable subsets
$\obj \subseteq \states^{\om}$ of runs starting at a given initial state.
By $\Prob[\mdp][\zstrat,\s]{\obj}$ we denote the
payoff under $\zstrat$, i.e., 
the probability that runs from $\s$ belong to $\obj$.
The value of a state is defined as
$\valueof[\mdp][\obj]{\s} \eqdef \sup_{\zstrat \in \zallstrats{\mdp}} \Prob[\mdp][\zstrat,\s]{\obj}$.
For $\eps > 0$ and state $\s$, a strategy $\zstrat \in \zallstrats{\mdp}$
is $\eps$-optimal iff $\Prob[\mdp][\zstrat,\s]{\obj} \ge \valueof[\mdp][\obj]{\s} - \eps$.
A $0$-optimal strategy is called \emph{optimal}.
An MD/MR strategy is called \emph{uniformly} $\eps$-optimal (resp.\ uniformly optimal)
if it is so from every start state.
An optimal strategy from $\s$ is called \emph{almost surely winning} if $\valueof[\mdp][\obj]{\s}=1$.
By $\AS\lrc{\obj}$ (resp.\ $\AS_{\finite}\lrc{\obj}$)
we denote the set of states that have an almost surely winning
strategy (resp.\ an almost surely winning finite-memory strategy)
for objective $\obj$. For ease of presentation, we drop subscripts and superscripts 
wherever possible if they are clear from the context.

We use the syntax and semantics of the LTL operators~\cite{CGP:book}
\mbox{$\eventually$ (eventually)}, $\always$ (always) and $\next$ (next)
to specify some conditions on runs.
% \textit{Reachability \& Safety.}
A reachability objective is defined by a set of target states
$\reachset \subseteq \states$. A run $\play = \s_0s_1 \ldots$
belongs to $\eventually\,\reachset$ iff $\exists i \in \nat\, \s_i \in \reachset$.
Similarly, $\play$ belongs to $\eventually^{\le n} \reachset$
(resp.\ $\eventually^{\ge n} \reachset$) iff
$\exists i \le n$ (resp.\ $i \ge n$) such that $\s_i \in \reachset$.
% $w$ satisfies the parametrised reachability objective $\reachn{n}{\reachset}$
% if such $i$ is $\leqslant n$, and positive reachability $\reachp{\reachset}$
% if $i > 0$. We also write $\eventually\;\reachset$, $\eventually^{\leqslant
%  n}\; \reachset$, $\next\eventually\;\reachset$ to denote the set of words
% $w$ that satisfy $\reach{\reachset}$, $\reachn{n}{\reachset}$,
% $\reachp{\reachset}$ respectively.
Dually, the safety objective
$\always\,\reachset$ consists of all runs
which never leave $\reachset$. We have $\always\,\reachset = \neg\eventually\neg\reachset$.

\begin{comment}
\textit{Parity.} A parity objective is defined via bounded function
$\coloring: \states \to \nat$ that assigns non-negative priorities
(aka colours) to states. Given an infinite run
$\play = \s_0s_1 \ldots$, let $\text{Inf}(\play)$
denote the set of numbers that occur infinitely often in the sequence
$\coloring(\s_0)\coloring(\s_1)\ldots$.
A run $\play$ satisfies \textit{even parity}
w.r.t.\ $\coloring$ iff the minimum of $\text{Inf}(\play)$ is even.
Otherwise, $\play$ satisfies \textit{odd parity}.
The objective even parity is denoted by $\Eparity(\coloring)$
and odd parity is denoted by $\Oparity(\coloring)$.
Most of the time, we implicitly assume that the colouring function is known
and just write $\Eparity$ and $\Oparity$.
Observe that, given any colouring $\coloring$, we have
$
\complementof{\Eparity} = \Oparity$ and
$
\Oparity(\coloring) = \Eparity(\coloring + 1) 
$
where $\coloring + 1$ is the function which adds $1$ to the colour of every
state. This justifies to consider only one of the even/odd parity objectives,
but, for the sake of clarity, we distinguish these objectives wherever necessary.
\end{comment}

\subparagraph{Energy/Reward/Counter-based objectives.}
Let $r: E \to \set{-R,\dots,0,\dots,R}$ be a bounded function that assigns rewards
to transitions. Depending on context, the sum of these rewards in a path
can be viewed as energy, cost/profit or a counter.
If $\s \movesto \s'$ and $r((\s,\s')) = c$, we write $\s \energymove{c}
\s'$. Let $\play = \s_0 \energymove{c_0} \s_1 \energymove{c_1} \ldots$ be a run.
We say that $\play$ satisfies
\begin{enumerate}
\item
  the $k$-\textit{energy} objective $\EN{k}$ iff $\lrc{k + \sum_{i=0}^{n-1} c_i} \geq 0$ for all $n \ge 0$.
\item
  the \emph{$l$-storage condition} $\infix\lrc{l}$ if $l+\sum_{i=m}^{n-1} c_i \ge 0$
holds for every infix $s_m \energymove{c_m} \s_{m+1}\ldots s_n$ of the run.
Let $\ES{k,l}$ denote the set of runs that satisfy both the $k$-energy
and the $l$-storage condition. Let $\ES{k} \eqdef \bigcup_l \ES{k,l}$. Clearly, $\ES{k} \subseteq \EN{k}$.
% Not used in this paper
% \item
%  \textit{Limit} objective $\LimInf\lrc{\rhd\, z}$ iff
%      $\lrc{\liminf_{n \tendsto \infty}\sum_{i=0}^{n-1} c_i } \rhd\; z$
%      for $\rhd \in \set{<,\le,=,\ge,>}$ and $z \in \R \cup \set{\infty,-\infty}$ and similarly for $\LimSup\lrc{\rhd\, z}$.
\item
      \textit{Mean payoff} $\MP\lrc{\rhd\,c}$ for some constant $c \in \R$
      iff $\lrc{\liminf_{n \tendsto \infty}\frac{1}{n}\sum_{i=0}^{n-1} c_i } \rhd\; c$
      for $\rhd \in \set{<,\le,=,\ge,>}$.
\end{enumerate}
A different way to consider the energy objective is to encode the energy level
(the sum of the transition weights so far) into the state space and
then consider the obtained infinite-state game with a safety objective.

An objective $\obj$ is called \textit{shift-invariant}
iff for all finite paths $\partialplay$ and plays $\play' \in \states^{\omega}$,
we have $\partialplay \play' \in \obj \iff \play' \in \obj$.
Mean payoff objectives are shift-invariant, but energy and
storage/infix objectives are not.
% Not used here
% Objective $\obj$ is called \textit{submixing} iff for all sequences of finite
% non-empty words $u_0$, $v_0$, $u_1$, $v_1 \ldots$ we have
% $u_0v_0u_1v_1 \ldots \in \obj \implies \left((u_0u_1\ldots \in \obj) \vee (v_0v_1 \ldots \in \obj)\right)$.

\subparagraph{Multidimensional reward-based objectives.} Let $\N,\Q,\R$ denote the set of positive integers, rationals and reals respectively. For a $d$-dimensional real vector $\vec{\mu}$, let $\mu_i$ 
denote the $i^{th}$ component of $\vec{\mu}$ for $1 \leq i \leq d$. Given two vectors $\vec{\mu},\vec{\nu} \in \R^d$, $\sim \in \set{<,\leq,>,\geq,=}$ we say $\vec{\mu} \sim \vec{\nu}$ if $\mu_i 
\sim \nu_i$ for every $i$. In particular, $\vec{\mu} > \vec{0}$ means that
\emph{every} component of $\vec{\mu}$ is strictly greater than $0$.
For a multidimensional reward function
$\vec{r}: \transition \to [-R,R]^d$,
we can consider any boolean combination of reward based objectives using any components of $\vec{r}$. For instance, 
$\obj_1 = \en_1(k) \, \cap \MP_2\lrc{> 0}$ denotes the objective that contains all runs that satisfy $\en(k)$ in the $1^{st}$ dimension and $\MP\lrc{> 0}$ in the $2^{nd}$ one.
We denote conjunctions of the same objective across different dimensions in
vectorized form, with the dimension information in the subscript. Therefore, % R#3: l 164.
$\vec{\en_{[a,b]}(\vec{k})} \, \cap \, \vec{\MP_{[c,d]}\lrc{> \vec{x}}}$
denotes the runs where the $\en_i(k_i)$ objective is satisfied for each $i \in [a,b]$
and the $\MP_j\lrc{> x_j}$ objective is satisfied for each $j \in [c,d]$. 
Given an infinite run $\play =  \s_0 \energymove{\vec{c_0}} \s_1 \energymove{\vec{c_1}} \ldots$,
let $X_n\lrc{\play} \eqdef \s_n$ denote the $n$-th state.
Let $\vec{Y_n}$ be the sum of the rewards in the first $n$ steps, i.e., $\vec{Y_n}\lrc{\rho} \eqdef \sum_{i=0}^{n-1} \vec{c_i}$.  % R#3: l 168.
These become random variables once an initial distribution and a strategy are fixed.

\begin{comment}
% 
% Also for the parameterized objective $\en(k)$ we write
% $\valueof[][\en]{\s,k}$ instead of $\valueof[][\en(k)]{\s}$
% and $\Prob{(\s,k)}(\en)$ instead of $\Prob{\s}(\en(k))$
% to better emphasize the dependence on $k$, and
% similarly for the objective $\Term(k)$.
\textbf{Energy-parity.}
We are concerned with approximating the value for the combined energy-parity
objective $\en(k) \cap \Eparity$ and building $\eps$-optimal strategies.

In our constructions we use some auxiliary objectives.
Following~\cite{MSTW2021}, these are defined as
$\Gain \eqdef \LimInf{>}{-\infty}\, \cap\, \Eparity$
and
$\Loss \eqdef \complementof{\Gain} = \LimInf{=}{-\infty}\, \cup\, \Oparity$.

\begin{remark}\label{rem:ssg-md}
For finite-state SSGs and the following objectives
there exist optimal MD strategies for both players.
Moreover, if the SSG is just a maximizing MDP then the set of states
that are almost surely winning for Maximizer can be computed in polynomial
time.
\begin{enumerate}
     \item $\eventually\,\reachset$ ~\cite{CONDON1992203}
     \item
       $\LimInf{\rhd}{-\infty}$,
       $\LimInf{\rhd}{\infty}$,
       $\LimSup{\rhd}{-\infty}$,
       $\LimSup{\rhd}{\infty}$,
       $\MP{>}{0}$ ~\cite[Prop.~1]{Brazdil2010}
     \item
       $\Eparity$ ~\cite{Zielonka:1998}
     \end{enumerate}
\end{remark}
\end{comment}

\subparagraph{Size of an instance.}
Given an MDP $\mdp = \mdptuple$ with reward function
$\vec{r}: \transition \to [-R,R]^d$, its size $|\mdp|$
is the number of bits used to describe it.
Similarly for $|\probp|$.
Transition probabilities and rewards can thus be stored in binary.
We call a size pseudo-polynomial in $|\mdp|$ if it is polynomial for the
case where $R$ is `small', i.e., if $R$ is given in unary.

\section{The Main Result}\label{sec:result}

\begin{theorem}\label{thm:main}
  Let $\mdp = \mdptuple$ be an MDP with $d$-dimensional rewards on the edges
  $\vec{r}: \transition \to [-R,R]^d$.
  For the multidimensional Energy-MeanPayoff objective
  $\en_1(k)\,\cap\,\vec{\MP_{[2,d]}\lrc{> 0}}$ the following properties hold.
  \begin{enumerate}
  \item
    The existence of an almost-surely winning strategy implies
    the existence of an almost-surely winning finite-memory strategy. \label{res:inf_str ==> fin_str}
  \item
    Moreover, a deterministic strategy with an exponential number of memory modes is
    sufficient. \label{res:mem_bound}
  \item
    An exponential (in $|P|$) number of memory modes is necessary in general,
    even for randomized strategies,
    even for
    $|\states|=5$, $d=2$ and $R=1$.
    \label{res:no_fixed_mem_suffices}
    \end{enumerate}
\end{theorem}

In the following three sections we prove items 1.,2.,3. of \Cref{thm:main},
respectively.

Here we sketch the main idea for the upper bound.
Except in a special corner case where the energy fluctuates only in a bounded
region, almost-surely winning strategies for Energy-MeanPayoff
can be chosen among some particular strategies 
that alternate between two modes, playing two different memoryless strategies.
This alternation keeps the balance between the Energy-part and the % # R3: l 187
MeanPayoff-part of the objective.
This is similar to almost-surely winning strategies for the Energy-Parity objective
in \cite{MSTW2017}.
In one mode, one plays a randomized memoryless strategy that almost surely yields a positive
mean payoff in all dimensions (in case of Energy-Parity, instead of mean payoff
it satisfies Parity almost surely). This is called the \emph{Gain} phase.
Whenever the energy level (the cumulative reward in dimension 1)
gets dangerously close to zero, one switches to the other mode and plays
a different memoryless strategy
that focuses exclusively on getting the energy level up again,
while temporarily neglecting the other part of the objective
(Parity or Mean payoff, respectively). This is called a \emph{Bailout}.
Once the energy level is sufficiently high, one switches back to the Gain
phase again.
The crucial property is that, except in a null set, only finitely many
Bailouts are required, and thus the temporary neglect of the second part of
the objective does not matter in the long run.
Such a strategy uses infinite memory, because it needs to remember the
unbounded energy level.
For Energy-Parity (and even Energy-co-B\"uchi) this cannot be avoided
and finite-memory strategies do not work \cite{MSTW2017}. 
However, for Energy-MeanPayoff one can relax the requirements somewhat.
Suppose that one records the stored energy only up to a certain bound $b$,
i.e., one forgets about potential excess energy above $b$.
In that case, one might have to do infinitely many Bailouts
with high probability, most of which are unnecessary (but one does not know
which ones). However, for a sufficiently large bound $b$, these superfluous
Bailouts occur so infrequently that they do not compromise the
MeanPayoff-part of the objective.
The critical part of the proof is to show this property and
an upper bound on $b$.
Once this is established, one obtains a finite-memory strategy,
because it suffices to record the energy level only in the %# R3: l 208
range $[0,b]$ (plus one extra bit of memory to record the current phase, Gain
or Bailout).

Note that the argument above is different from the one that justifies
finite-memory $\eps$-optimal strategies for \emph{Energy-Parity} in
\cite{Dantam-Mayr:LIPIcs.MFCS.2023.38}.
These also record the energy only in a bounded region, but stop doing Bailouts
after the upper bound has been visited.
I.e., they do too few Bailouts, and thus incur an
$\eps$-chance of losing.
In contrast, our almost-surely winning strategies for 
Energy-MeanPayoff rather do too many Bailouts, but sufficiently infrequently
such that they don't compromise the objective.

\section{Proof of \texorpdfstring{\Cref{res:inf_str ==> fin_str}}{finite memory strategy existence}} \label{sec:inf_str ==> fin_str}

W.l.o.g, we assume that every state in
$\mdp$ has an almost surely winning strategy for Energy-MeanPayoff for some
initial energy level. (Otherwise, consider a suitably restricted sub-MDP.)
For conciseness, we denote the objective by
$\obj\lrc{k} \eqdef \en_1(k)\,\cap\,\vec{\MP_{[2,d]}\lrc{> 0}}$.
Let 
$$\winset\lrc{\s} \eqdef \setcomp{k}{\s \in \AS\lrc{\obj\lrc{k}}},\quad i_{\s} \eqdef \min\lrc{\winset\lrc{\s}}$$ 
denote the possible initial energy levels and the minimum initial energy level
such that one can win almost surely from state $\s$.
In particular, $i_{\s}$ is well defined by our assumption on $\mdp$.

Towards a contradiction, 
assume that not all configurations are winnable with a finite-memory strategy.
I.e., let $\winset_{\finite}\lrc{\s} \eqdef \setcomp{k}{\s \in
  \AS_f\lrc{\obj\lrc{k}}}$ denote the energy levels 
from which one can win almost surely with a \emph{finite-memory} strategy from $\s$,
and assume that there is a state $\s^{\dagger}$ such that 
$i_{\s^\dagger} \notin \winset_{\finite}\lrc{\s^{\dagger}}$.
We then construct a finite-memory winning strategy from $\s^{\dagger}$
for $\obj\lrc{i_{\s^\dagger}}$, leading to a contradiction.
Similar to $i_\s$, let $f_\s$ denote the 
minimal $k$ such that $k \in \winset_{\finite}\lrc{\s}$ and $\infty$ if there
is no such $k$.

\begin{definition}
We construct a new MDP $\mdp^\ast$ which abstracts away all the $\winset_{\finite}$ configurations. 
At every state $\s$, the player gets the option to enter a
winning sink state if the energy level is sufficiently large to win with
finite memory, i.e., if the current energy level is at least $f_\s$.
The states of the MDP $\mdp^\ast$ will have two copies of each state $\s$ of
$\mdp$, namely $\s$ and $\s^\prime$. Moreover, we add a new state 
$\s_\textrm{win}$.
All states $\s^\prime$ are controlled by $\pz$ and every step
$\s_1 \movesto \s$ in the original MDP $\mdp$ is now mapped to a step
$\s_1 \movesto \s^\prime$ with the same reward
(and the same probability if $\s_1$ was a random state).
In $\s^\prime$, the player has two choices: he can either go to $\s$ with
reward $\vec{0}$ or go to $\s_\textrm{win}$ with reward
$\tuple{-f_\s,\vec{0}}$. The latter choice is only available if $f_\s < \infty$.
$\s_\textrm{win}$ is a winning sink where
$\s_{\textrm{win}} \movesto \s_{\textrm{win}}$ with reward $\vec{1}$, i.e.,
reward $+1$ in all dimensions.
\end{definition}

% Formally, 
%$\mdp^\ast \eqdef \tuple{\states^\ast,\zstates^\ast,\rstates^\ast,\transition^\ast,\probp^\ast}$ and $\vec{r}^\ast : \transition^\ast \to [-R^\ast,R] \times [-R,R]^{d-1}$ where
%\begin{itemize}
%    \item $\rstates^\ast \eqdef \rstates$
%    \item $\zstates^\ast \eqdef \zstates \cup \states^{\prime} \cup \{\s_\textrm{win}\}$
%    \item $\states^\ast \eqdef \rstates^\ast \cup \zstates^\ast$
%    \item For every edge $e = (\s_1\step{}\s) \in \transition$,
%      we have edges $e_1=(\s_1\step{}\s^\prime)$ and
%      $g_{\s}=(\s^\prime \movesto \s)$.
%      Moreover, if $f_\s < \infty$, we also have an edge
%      $h_{\s}=(\s^\prime \movesto \s_{\textrm{win}})$ in $\transition^\ast$,
%      where $\vec{r}^\ast\lrc{e_1,g_\s,h_{\s}} = \tuple{\vec{r}\lrc{e},\vec{0},\lrc{-f_{\s},\vec{0}}}$, $\probp^\ast\lrc{e_1} = \probp\lrc{e}$ if $\s_1 \in \rstates$
%    \item $\s_{\textrm{win}} \movesto \s_{\textrm{win}} \in \transition^\ast$ with reward $\vec{1}$.
% \end{itemize}

The following lemma shows that the existence of almost surely winning (finite-memory) strategies
coincides in $\mdp^\ast$ and $\mdp$.

\begin{lemma}\label{lem:mdp_equiv_mdp^ast}
  Let $\s \in \states$ and $k \in \N$, and consider the objective $\obj\lrc{k}$.
  There exists an almost surely winning
  strategy $\zstrat^\ast$ from $\s$ in $\mdp^\ast$ if and only if
  there exists an almost surely winning
  strategy $\zstrat$ from $\s$ in $\mdp$.
  Moreover, if $\zstrat^\ast$ is finite-memory then $\zstrat$ can be chosen as
  finite-memory, and vice-versa.
\end{lemma}
\begin{proof}
  Towards the `only if' direction, let $\zstrat^\ast$ be a strategy from $\s$
  in $\mdp^\ast$ that is almost surely winning for $\obj\lrc{k}$.
  We define a strategy $\zstrat$ from $\s$ in $\mdp$ that plays as follows.
  First $\zstrat$ imitates the moves of $\zstrat^\ast$ until (if ever)
  $\zstrat^\ast$ chooses a move $\s^\prime_1 \to \s_\textrm{win}$ with non-zero
  probability at some state $\s^\prime_1$.
  This is possible, since any finite path in $\mdp^\ast$ that does not contain $\s_\textrm{win}$
  can be bijectively mapped to a path in $\mdp$. The only difference is that
  paths in $\mdp^\ast$ contain extra steps via primed states, which are
  skipped in the paths in $\mdp$. Moreover, the transition probabilities at
  random states coincide in $\mdp^\ast$ and $\mdp$.
  If $\zstrat^\ast$ chooses a move $\s^\prime_1 \to \s_\textrm{win}$ with non-zero
  probability at some state $\s^\prime_1$ then the current energy level must
  be $\geq f_{\s_1}$, because $\zstrat^\ast$ satisfies the energy objective
  almost surely (and thus even surely). 
  Thus, in $\mdp$, there exists an almost surely winning finite-memory 
  strategy $\hat{\zstrat}$ for $\obj\lrc{f_{\s_1}}$ from $\s_1$.
  In this situation $\zstrat$ continues by playing $\hat{\zstrat}$ from $\s_1$.
  Therefore, $\zstrat$ satisfies the energy objective surely.
  Moreover, by shift invariance and the properties of $\hat{\zstrat}$,
  it also satisfies the Mean payoff objective almost surely.
  Thus, $\zstrat$ satisfies $\obj\lrc{k}$ almost surely.
Finally, if $\zstrat^\ast$ is finite-memory then so is $\zstrat$,
because $\hat{\zstrat}$ is also finite-memory.

Towards the `if' direction, let $\zstrat$ be a strategy from $\s$
in $\mdp$ that is almost surely winning for $\obj\lrc{k}$.
We define a strategy $\zstrat^\ast$ from $\s$ in $\mdp^\ast$ that 
imitates the moves of $\zstrat$. Moreover, at primed states
$q'$ it always goes to $q$ (and never to $\s_\textrm{win}$).
Since the probabilities at random states coincide in $\mdp^\ast$ and $\mdp$,
also the probabilities of the induced paths coincide.
The only difference is that the runs in $\mdp^\ast$ contain extra steps via
primed states and these extra steps carry reward zero.
Thus, the mean payoff of a run in $\mdp^\ast$ is $1/2$ the mean payoff of the
corresponding run in $\mdp$. However, this does not affect the property that
the mean payoff is $>0$ almost surely in either MDP.
Thus, $\zstrat^\ast$ satisfies $\obj\lrc{k}$ almost surely.
Finally, if $\zstrat$ is finite-memory then so is $\zstrat^\ast$.
\end{proof}

The next lemma shows that, in $\mdp^{\ast}$, it is impossible to
satisfy Energy-MeanPayoff from $\s$ with arbitrarily high probability,
unless one also allows arbitrarily large fluctuations in the energy level,
or $f_\s = i_\s$. (Recall that $f_\s, i_\s$ are defined relative to $\mdp$.)

\begin{lemma}\label{lem:fluctuate}
For every state $\s$ with $f_\s > i_\s$ and every
$\ell \in \N$, there exists a $\delta_\ell > 0$ such that
$\valueof[\mdp^\ast][\obj\lrc{i_{\s}} \, \cap \, \infix_1\lrc{\ell}]{\s}
\leq 1 - \delta_\ell$.
\end{lemma}
\begin{proof}
  Towards a contradiction, assume
  that $\valueof[\mdp^\ast][\obj\lrc{i_{\s}} \, \cap \, \infix_1\lrc{\ell}]{\s} = 1$ for some $\ell$.

 $\obj\lrc{i_{\s}}\, \cap \, \infix_1\lrc{\ell} =
 \en_1(i_{\s})\,\cap\,\vec{\MP_{[2,d]}\lrc{> 0}} \, \cap \, \infix_1\lrc{\ell}
 = \ST_1\lrc{i_{\s},\ell}\, \cap \, \vec{\MP_{[2,d]}\lrc{> 0}}$.
 Therefore, we have $\valueof[\mdp^\ast][\s]{\ST_1\lrc{i_{\s},\ell}\, \cap \, \vec{\MP_{[2,d]}\lrc{> 0}}} = 1$.
 Below we prove that this objective has a finite-memory almost-surely winning strategy $\zstrat$ in $\mdp^\ast$. % R#1 : l.271
 Consider a modified MDP $\mdp^\ast_1$ that encodes the 
 energy level up to $i_\s + \ell$ in the states.
 A step exceeding the upper energy bound $i_\s + \ell$ results in a truncation
 to $i_\s + \ell$, while a step leading to a negative energy leads to a losing sink.
 There exists a memoryless randomized (MR) strategy $\zstrat_1$ in
 $\mdp^\ast_1$ from state $(\s, i_\s)$
 that wins $\vec{\MP_{[2,d]}\lrc{> 0}}$ almost surely, by \Cref{lem:mr_gain_strategy}. % R#3: l.274
 We can then carry $\zstrat_1$ back to $\mdp^\ast$ as a finite-memory strategy $\zstrat$
 with $i_\s + \ell + 1$ memory modes such that
 $\Prob[\mdp^\ast][\zstrat,\s]{\ST_1\lrc{i_{\s},\ell}\, \cap \, \vec{\MP_{[2,d]}\lrc{> 0}}} = 1$.
 By set inclusion, $\Prob[\mdp^\ast][\zstrat,\s]{\obj\lrc{i_{\s}}} = 1$.
 By \Cref{lem:mdp_equiv_mdp^ast}, there also exists
 a finite-memory strategy from $\s$ in $\mdp$ that is almost surely winning for
 $\obj\lrc{i_{\s}}$.
 This implies $f_\s = i_\s$, a contradiction to our assumption $f_\s > i_\s$.
 Hence, we obtain $\delta_{\ell} \eqdef 1-
 \valueof[\mdp^\ast][\obj\lrc{i_{\s}} \, \cap \, \infix_1\lrc{\ell}]{\s} > 0$.
\end{proof}

The following three lemmas show that almost surely winning strategies for
Energy-MeanPayoff can be found by combining two different
memoryless strategies for the simpler $\Bailout$ and $\Gain$ objectives.

First, we define the objective $\Bailout(k) \eqdef \en_1(k) \, \cap\, \MP_1\lrc{> 0}$.
Let $i^{\Bailout}_{\s}$ denote the minimal energy value $k$ with which one can
almost surely satisfy $\Bailout(k)$ when starting from
state $\s$ (or $\infty$ if it does not exist).

\begin{lemma}~\cite[Lemma 3]{BKN2016} \label{lem:md_bailout_strategy}
Let $\mdp$ be an MDP.
If $\s \in \AS(\Bailout(k))$ for some $k \in \N$   
then $i^{\Bailout}_{\s} \leq 3 \cdot \size{\mdp} \cdot R$.
Moreover, there exists a uniform MD strategy $\optzstrat_{\Bailout}$
which is almost surely winning $\Bailout(k)$
from every state $\s \in \AS(\Bailout(k))$.
% (which implies $k \geq i^{\Bailout}_{\s}$).
\end{lemma}
%\TODO{Shouldn't the existence of a memoryless strategy further constraint the safe level to $ \leq \size{\mdp} \cdot R$}

We define the $\Gain$ objective as $\MP_{[1,d]}\lrc{> 0}$.
The following lemma shows that an almost surely winning strategy
$\optzstrat_{\Gain}$ for this objective can be chosen as memoryless randomized.

\begin{lemma}~\cite[Proposition 5.1]{brazdil2014markov}\label{lem:mr_gain_strategy}
  There is a uniform MR strategy $\optzstrat_{\Gain}$ which is almost surely
  winning for $\Gain$ (or any subset of dimensions) from % R#3: l.274
  all states $\s \in \AS(\Gain)$.
\end{lemma}

A difference between $\mdp^\ast$ and $\mdp$ is that if one can 
almost surely win Energy-MeanPayoff in $\mdp^\ast$ then one can
also push the energy level arbitrarily high.
This does not always hold in $\mdp$.
(Consider, e.g., a single-state Markov chain with a single loop with reward $0$
in the $1^{st}$ dimension and $+1$ in all other dimensions.)
The difference comes from the loop at state $\s_{\textrm{win}}$ in $\mdp^\ast$
which has a strictly positive reward in all dimensions.
Thus, the following lemma only holds for $\mdp^\ast$.

\begin{lemma}\label{lem:gain_bailout_existence}
  In $\mdp^{\ast}$,
    there are two uniform memoryless strategies 
    $\optzstrat_{\Bailout}$ and $\optzstrat_{\Gain}$
    which,
    starting from any state $\s\in \AS(\obj\lrc{k})$, 
    almost surely satisfy $\Bailout(k)$ 
    and $\Gain$, respectively.
\end{lemma}
\begin{proof}
Let $\s\in \AS(\obj\lrc{k})$.  
We show that $\s \in \AS(\Bailout(k))$ and
$\s \in \AS\lrc{\Gain}$. %  R#1 : l.301 
The existence of the memoryless strategies $\optzstrat_{\Bailout}$ and $\optzstrat_{\Gain}$
then follows from \Cref{lem:md_bailout_strategy} and
\Cref{lem:mr_gain_strategy}, respectively.

We assumed that all states $\s$ in $\mdp$ admit an almost surely winning strategy
for Energy-MeanPayoff.
By \Cref{lem:mdp_equiv_mdp^ast}, this also holds for all states $q$ in
$\mdp^\ast$.
Let $\zstrat^\sharp_q$ denote an almost surely winning strategy from $q$
for $\obj\lrc{i_q}$ in $\mdp^\ast$
(without restrictions on memory).

Recall from \Cref{sec:prelim} that the random variable $X_t$ denotes the state %  R#3: l.307
at time $t$, and $Y_t$ denotes the ($d$-dimensional) sum of the rewards until
time $t$.

\begin{claim}\label{lem:Hit_swin_or_pos_gain_with_pos_prob}
  For every state $q \in \mdp^\ast$ there exists some number of
  steps $n_q \in \N$ and a probability $p_q > 0$ such that
        \[
          \Prob[\mdp^\ast][\zstrat^\sharp_q,q]{\bigcup_{j=0}^{n_q} ((Y_j)_1 > i_{X_j} - i_q)\, \cup\, ((Y_j)_1 \geq f_{X_j} - i_q)} \geq p_q.
        \]
    \end{claim}
    \begin{claimproof}%  R#1 : l.312
      Towards a contradiction, assume that for all $m$
      \[
        \Prob[\mdp^\ast][\zstrat^\sharp_q,q]{\bigcup_{j=0}^{m} ((Y_j)_1 >
          i_{X_j} - i_q)\, \cup\, ((Y_j)_1 \geq f_{X_j} - i_q)}=0.
        \]
        Due to the second part of the union, this implies that never
        $(Y_j)_1 + i_q \geq f_{X_j}$.
        Since $\zstrat^\sharp_q$ satisfies $\en_1(i_q)$ almost surely, it can
        never choose the step to $\s_\textrm{win}$.
        This implies $\Prob[\mdp^\ast][\zstrat^\sharp_q,q]{\eventually \s_\textrm{win}}= 0$, i.e., $X_j$ is always different from
        $\s_\textrm{win}$.  
        (The values $f_\s$ were initially defined with respect to states $\s$ of the original MDP $\mdp$,
        but the definition is naturally extended to the MDP $\mdp^\ast$, by
        giving the primed states the same value, i.e., $f_{\s'} = f_{\s}$.
        The state $\s_\textrm{win}$ does not appear in $\mdp$, but only in $\mdp^\ast$.
        We can extend the definition by having $f_{\s_\textrm{win}}=0$.
        However, this is not strictly required. The $f_{X_j}$ is already
        defined, since $X_j$ is always different from $\s_\textrm{win}$.)
                
        Since $\zstrat^\sharp_q$ satisfies $\en_1(i_q)$ almost surely, %  R#3: l.317
        all runs always satisfy $(Y_j)_1 \geq i_{X_j} - i_q$ for all $j$.
        On the other hand, our assumption yields
        $\Prob[\mdp^\ast][\sigma^\sharp_q,q]{\bigcup_{j=0}^{m} (Y_j)_1 > i_{X_j} - i_q} = 0$.
        This implies that $(Y_j)_1 = i_{X_j} - i_q$ for all $j$.
        Hence, in all runs the energy fluctuates by at most $\ell \eqdef 2\max_q i_q$.
        Thus, $\Prob[\mdp^\ast][\zstrat^\sharp_q,q]{\obj\lrc{i_{q}} \, \cap \, \infix_1\lrc{\ell}}=1$.
        Then \Cref{lem:fluctuate} implies that $f_q = i_q$. Since $X_0 = q$ we
        have $f_{X_0} = f_q$ and thus $(Y_0)_1 \geq f_{X_0} - i_q = 0$.
        This contradicts our assumption, since the second part of the union is
        surely satisfied.
    \end{claimproof}%  R#1 : l.312

    For any state $q$, let $n_q$, $p_q$ denote the values from
    \Cref{lem:Hit_swin_or_pos_gain_with_pos_prob}.
    
    Now we show that $\s \in \AS(\Bailout(k))$.
    Define a strategy $\zstrat_{\Bailout}$ which plays in phases, %  R#1 : l.324
    separated by resets.
    It remembers the number of steps $t \geq 0$ since last reset, %  R#3 : l.325
    the (under-approximated) sum of rewards $Q_t$
    and the current state $X_t$.
    The first phase starts at state $\s$ and $\zstrat_{\Bailout}$ plays like $\zstrat^\sharp_\s$ 
    until one of the following events occur.
    \begin{enumerate}
    \item
      There is enough energy such that it is safe to move to
      $\s_\textrm{win}$, i.e., $\lrc{Q_t \geq f_{X_t} - i_\s}$, or \label{itm:Hit_swin}
    \item
      The current energy level is strictly greater than the minimal required
      energy level of the current state, i.e., $\lrc{Q_t > i_{X_t} - i_\s}$, or \label{itm:Pos_gain}
    \item
      $n_\s$ steps have elapsed, i.e., $\lrc{t = n_\s}$. \label{itm:timeout}
    \end{enumerate}
    If at any point \Cref{itm:Hit_swin} happens, then the strategy simply goes
    to $\s_\textrm{win}$.
    If it is the case that \Cref{itm:Pos_gain} occurs before $t = n_\s$, let's
    say at some time $t'$, then the phase ends at $t'$. The sum of the rewards
    in the phase, between the last reset (where $t=0$) and the current time is $\geq i_{X_{t'}}-i_\s + 1$. % R#1 : l.334
    If neither \Cref{itm:Hit_swin} nor \Cref{itm:Pos_gain} occurs before
    $t=n_\s$,
    then the phase ends and we let $t' \eqdef t =n_\s$.
    The sum of the rewards in this phase is then exactly $i_{X_{t'}} - i_\s$.
    At the end of the phase $\zstrat_{\Bailout}$ resets the number of steps
    $\lrc{t = 0}$,
    and $Q_t$ to $0$.
    In the following phase it moves according to $\zstrat^\sharp_{X_{t'}}$ until the next reset.

    $\zstrat_{\Bailout}$ clearly satisfies $\en_1(k)$ as it is a mix of energy
    safe strategies $\lrc{\zstrat^{\sharp}_q}_{q \in \states^{\ast}}$
    and since we are starting from a safe energy level.
    By \Cref{lem:Hit_swin_or_pos_gain_with_pos_prob}, there is a positive
    probability (lower-bounded by $\min_q p_q >0$)
    that either \Cref{itm:Hit_swin} or \Cref{itm:Pos_gain} happens in each phase.
    
    Hence, unless event \Cref{itm:Hit_swin} occurs, \Cref{itm:Pos_gain} occurs
    infinitely often almost surely. Moreover, since the length of phases is
    upper bounded by $\max_q n_q$, it occurs frequently. %  R#1 : 343
    We obtain $\Prob[\mdp^\ast][\zstrat_{\Bailout},\s]{\MP_1 \geq
      \min_{q}\lrc{\frac{p_q}{n_q}} >0 \given \neg\eventually \s_\textrm{win}} = 1$.
    On the other hand, if $\s_\textrm{win}$ is reached, then $\MP_1$ holds by
    shift invariance and the definition of the positive rewards in the loop at
    $\s_\textrm{win}$.
    Therefore, $\Prob[\mdp^\ast][\zstrat_{\Bailout},\s]{\en_1(i_\s) \, \cap \, \MP_1\lrc{> 0}} = 1$.

    Now we show that $\s \, \in \, \AS\lrc{\Gain}$.
    We make use of the following strategies.
    \begin{itemize}
        \item $\zstrat^\sharp_q$ which satisfies
          $\en_1(k)\,\cap\,\vec{\MP_{[2,d]}\lrc{> 0}}$ almost surely from $q$ for every $k \geq i_q$.
        \item a uniform MD strategy $\optzstrat_{\MP_1}$ which satisfies
          $\MP_1\lrc{> 0}$ almost surely from every state. It exists since
          $\AS\lrc{\MP_1\lrc{> 0}} = \states^{\ast}$
          (where $\states^{\ast}$ is the set of states of $\mdp^\ast$),
          because $\Prob[\mdp^\ast][\zstrat_{\Bailout},\s]{\en_1(i_\s) \, \cap \, \MP_1\lrc{> 0}} = 1$.
    \end{itemize}

    From the former, we get probabilistic bounds on the achievable mean payoff
    in all the dimensions, \ie for all states $\s$, and $0 \leq \eps < 1$,
    there is a $d-1$ dimensional vector $\vec{\nu_{\eps}} > \vec{0}$ such that
    $\Prob[\mdp^\ast][\zstrat^\sharp_\s,\s]{\vec{\MP_{[2,d]}} \geq \vec{\nu_{\eps}}}
    \geq 1-\frac{\eps}{2}$.
    This follows from the fact that for any sequence of decreasing vectors
    $\vec{\nu_n} \tendsto \vec{0}$ in $\R^{d-1}$, $\vec{\MP_{[2,d]}\lrc{> 0}}
    = \bigcup_n \vec{\MP_{[2,d]}\lrc{\geq \vec{\nu_n}}}$ and continuity of
    measures. % R#1 : l.354
    Furthermore, denoting by $\vec{Y_t}$ the sum of rewards in all dimensions
    until time $t$,
    there exists a sufficiently large bound $n_{\eps} \in \N$
    such that $\Prob[\mdp^\ast][\zstrat^\sharp_\s,\s]{\frac{\lrc{Y_t}_j}{t}
      \geq \frac{\lrc{\nu_{\eps}}_j}{2}} \geq 1-\eps$
    in each of the dimensions $j \in [2,d]$ for all $t \geq n_{\eps}$ steps.
    This can be shown by observing that $\MP_j\lrc{\geq \lrc{\nu_{\eps}}_j} =
    \bigcap_{k=1}^{\infty}\bigcup_{n=1}^{\infty} \bigcap_{t=n}^{\infty}
    \lrc{\frac{\lrc{Y_t}_j}{t} \geq \lrc{\nu_{\eps}}_j \cdot
      \lrc{1-\frac{1}{2^k}}}$ and using continuity of measures.

    Similarly, there exists a bound $n^*_{\eps} \in \N$ and value $\nu^*_{\eps}>0$
    such that $\Prob[][\optzstrat_{\MP_1},\s]{\frac{\lrc{Y_t}_1}{t} \geq
      \frac{\nu^*_{\eps}}{2}} \geq 1-\eps$ after $t \geq n^*_{\eps}$ steps for
    every state $\s$.

    Now consider the following strategy $\zstrat_{\Gain}$, which switches between two phases.
    \begin{description}
        \item[Phase 1:] If the current state is $q$, it moves according to
          $\zstrat^\sharp_q$ for
          some number $\alpha > n_{\eps}$ of steps. Then it switches to phase 2.
        \item[Phase 2:] It moves according to $\optzstrat_{\MP_1}$ for some
          number $\beta > n^*_\eps$ of steps, and then switches back to phase 1.
        \end{description}
    The strategy $\zstrat_{\Gain}$ is a finite-memory strategy, since the
    lengths of the alternating phases are bounded by $\alpha$ and $\beta$,
    respectively. (Even if $\zstrat^\sharp_q$ is an infinite-memory strategy,
    it can only use bounded memory in each phase.)
    
    We fix $\zstrat_{\Gain}$ from the start state $\s$ and obtain a
    finite-state Markov chain. In every BSCC of this Markov chain, the expected mean payoff in the $1^{st}$ dimension will be 
    $$\geq \frac{-i^\sharp + \beta \cdot \lrc{1-\eps} \cdot \lrc{\frac{\nu^*_{\eps}}{2}} - \beta \cdot \eps \cdot R }{\alpha + \beta}.$$
    where $i^\sharp = \max_\s i_\s$ denotes the maximum (over all states) minimal safe energy. 

    Similarly, in every BSCC, the expected mean payoff in the $j^{th}$ dimension for $j \geq 2$ can be lower-bounded by 
    $$\geq \frac{\alpha \cdot \lrc{\lrc{1-\eps} \cdot \lrc{\frac{\lrc{\nu_{\eps}}_j}{2}}-\eps \cdot R} - \beta \cdot R}{\alpha + \beta}.$$

    By choosing $\eps$ sufficiently small, $\beta$ sufficiently large to make
    the first term positive and $\alpha \gg \beta$ sufficiently large to make
    the second term positive, we can get positive expected mean payoff in all
    dimensions. Since this holds in every BSCC of the induced finite Markov
    chain, the objective $\Gain$ is satisfied almost surely.
\end{proof}

The following lemma shows the converse of \Cref{lem:gain_bailout_existence}.
In $\mdp^{\ast}$, it is always possible to win $\obj\lrc{i_{\s}}$ almost
surely from $\s$ by playing a particular strategy $\optzstrat_{\mathtt{alt,Z_b,Z_g}}$
which combines the two uniform memoryless strategies 
$\optzstrat_{\Bailout}$ and $\optzstrat_{\Gain}$.
Let $Z_b$ denote the minimal universally safe energy level for
$\Bailout$, i.e., $Z_b \eqdef \max_\s \min\{k\mid \s \in \AS(\Bailout(k))\}$.
Moreover, let $Z_g > Z_b$ be a larger energy level at which our strategy
switches from $\optzstrat_{\Bailout}$ to $\optzstrat_{\Gain}$.

\begin{comment}
$\eqdef \memstrattuple$ where
\begin{itemize}
    \item $\memconfset = \N \x \{0,1\}$ where the first component is the actual energy level and an additional bit which denotes if we are in `$\Gain$-phase' or `$\Bailout$-phase'
    \item  $\initmem = \tuple{k,0}$ if $k < k^{\max} +  R$ else $\tuple{k,1}$
    \item $\memup\lrc{\tuple{k,y},e} = \tuple{k_{\mathtt{new}},x}$ where \begin{itemize}
        \item $k_{\mathtt{new}} = k+\vec{r}^\ast(e)_1$
        \item $x = \begin{rcases}
            k_{\mathtt{new}} \geq Z & \textrm{ if } y=0 \\
            k_{\mathtt{new}} \geq k^{\max} + R & \textrm{ if } y = 1
        \end{rcases}$
    \end{itemize}
    \item $\memsuc\lrc{\tuple{k,x},\s} = \optzstrat_{\Bailout}\lrc{\s}$ if $x = 0$, $\optzstrat_{\Gain}\lrc{\s}$ o.w  
\end{itemize}
\end{comment}

Similarly to \cite{MSTW2017}, we define an infinite-memory strategy $\optzstrat_{\mathtt{alt,Z_b,Z_g}}$
that always records the current energy level and operates by switching between two phases.
It starts by playing $\optzstrat_{\Gain}$ (Gain-phase)
if our starting energy level is sufficiently high $(\geq Z_b + R)$,
and otherwise starts by playing $\optzstrat_{\Bailout}$ (Bailout-phase).
In the $\Bailout$-phase, the primary goal is to pump the energy level up until
it is $\ge Z_g$, and then it switches to the $\Gain$-phase.
It enters the $\Bailout$-phase again if the energy level drops below $Z_b + R$
(in which case it will still be $\ge Z_b$).

\begin{lemma}\label{lem:reverse_gain_bailout_existence}
  There exists a $Z_g \in \N$ such that for every $\s$ in $\mdp^*$
  the strategy $\optzstrat_{\mathtt{alt,Z_b,Z_g}}$ is almost surely
winning for $\obj\lrc{i_{\s}}$ from $\s$.
\end{lemma}
\begin{proof}
The parameter $Z_g$ is chosen sufficiently large such that there is a fixed non-zero
probability that after every Bailout-phase one never needs another $\Bailout$. % R#3: l.389
(Thus, except in a null set there are only finitely many Bailouts.) % R#3: l.390
The existence of such a finite $Z_g$ is guaranteed by the fact that
$\lim_{k \tendsto \infty} \Prob[][\optzstrat_{\Gain},\s]{\obj\lrc{k}} = 1$.~(\Cref{lem:exp_decay}).
Eventually, except in a null set, $\optzstrat_{\mathtt{alt,Z_b,Z_g}}$ plays
$\Gain$ forever, thus satisfying $\obj\lrc{i_{\s}}$ almost
surely from $\s$.
\end{proof}

Some combined objectives like Energy-Parity really require infinite
memory for almost surely winning strategies \cite{MSTW2017}.
However, we show that a sufficiently large \emph{finite} memory is enough to
win Energy-MeanPayoff almost surely. 
The idea is to modify the strategy $\optzstrat_{\mathtt{alt,Z_b,Z_g}}$ such
that it remembers the current energy only in the interval $[0,b]$, for some
sufficiently large $b > Z_g$, and ignores any possible excess energy above $b$.
This modified strategy is denoted by $\optzstrat_{\mathtt{alt,Z_b,Z_g,b}}$, and
it has a finite set of memory modes $[0,b] \x \{0,1\}$. The $\{0,1\}$ part
is used to remember the current phase ($\Gain=0$ or $\Bailout=1$).
Then $\optzstrat_{\mathtt{alt,Z_b,Z_g,b}}[(u,x)]$ denotes the strategy
$\optzstrat_{\mathtt{alt,Z_b,Z_g,b}}$ with current memory mode $(u,x) \in [0,b] \x \{0,1\}$. 

The finite bound $b$ on the remembered energy has the effect that
$\optzstrat_{\mathtt{alt,Z_b,Z_g,b}}$ can no
longer guarantee a fixed positive probability of not needing another Bailout
after each Bailout-phase. Thus, one might have infinitely many Bailouts with
positive probability. (Most of these are unnecessary, but one cannot be sure
which ones).
Unlike for Energy-Parity, where using infinitely many 
$\Bailout$ phases can compromise the objective,
the nature of the $\vec{\MP_{[2,d]}\lrc{> 0}}$ objective allows us to use
infinitely many Bailouts with non-zero probability, provided that they happen
sufficiently infrequently.

By its construction, the strategy $\optzstrat_{\mathtt{alt,Z_b,Z_g,b}}[(i_\s,x)]$ is
energy-safe from every state $\s$, every initial energy $\ge i_s$ and $x \in \{0,1\}$.
It remains to show that it 
also satisfies $\vec{\MP_{[2,d]}\lrc{> 0}}$ almost surely.
Since $\optzstrat_{\mathtt{alt,Z_b,Z_g,b}}$ is finite-memory, it suffices to
consider the induced finite Markov chain $\mc$ and show that the expected mean
payoff is strictly positive in every BSCC.
I.e., we prove that
$\expectation[][\optzstrat_{\mathtt{alt,Z_b,Z_g,b}},\s]{\vec{\MP_{[2,d]}}} >\vec{0}$
  for a sufficiently large $b$.
To this end, we consider the finite Markov chains   
$\mc^{\Gain}$ and $\mc^{\Bailout}$ obtained by fixing the memoryless
strategies $\optzstrat_{\Gain}$ and $\optzstrat_{\Bailout}$ in $\mdp^\ast$,
respectively.
The application of $\optzstrat_{\mathtt{alt, Z_b, Z_g, b}}$ can then be seen
as alternating between these two Markov chains based on hitting certain energy
levels. 

Let $T^{\Gain}$ denote the random variable that measures the length of a
Gain-phase, when starting at energy level $Z_g$ and assuming that the energy
it truncated at $b$.
Similarly, $T^{\Bailout}$ is the random variable that measures the length of a
Bailout-phase when starting at energy level $Z_b$.
(Here it does not matter that the energy is truncated at $b$,
since the Bailout-phase ends when the energy reaches $Z_g < b$.)
Since $R$ can be $>1$, the Bailout-phase might actually start at a slightly higher
energy level $u \in [Z_b,Z_b+R-1]$, and thus $T^{\Bailout}$ over-approximates the
actual length of the Bailout-phase, which is conservative for our analysis.
Similarly, 
the Gain phase might start with an energy slightly higher than $Z_g$, and
$T^{\Gain}$ under-approximates the length of the Gain-phase, which is again
conservative.
The random variables $(Y_{T^{\Gain}})_i$ and $(Y_{T^{\Bailout}})_i$
then measure the sum of the rewards the $i^{th}$ dimension obtained during the
Gain and Bailout phases, respectively.

The following lemma shows that the strategy $\optzstrat_{\mathtt{alt, Z_b, Z_g, b}}$
can attain a strictly positive mean payoff in all dimensions $i \in [2,d]$,
provided that the expected reward during the Gain-phase is sufficiently large
(positive) and the expected reward during the Bailout-phase (though possibly
negative) is not too small.

\begin{lemma}\label{lem:sufficiency_condition_pos_mp}
  If there are constants $v^1_i>0$ and $v^2_i$ such that, for all $i \in [2,d]$ and
  states $q$
    \begin{align*}
        \expectation[\mdp^\ast][\optzstrat_{\mathtt{alt, Z_b,Z_g,b}}[(Z_g,0)],q]{\lrc{Y_{T^{\Gain}}}_i} &\geq v^1_i\\% R#3: l.434
        \expectation[\mdp^\ast][\optzstrat_{\mathtt{alt, Z_b, Z_g, b}}[(Z_b,1)],q]{\lrc{Y_{T^{\Bailout}}}_i} &\geq v^2_i\\% R#3: l.435
        v^1_i + v^2_i > 0
    \end{align*}
    then $\expectation[\mdp^\ast][\optzstrat_{\mathtt{alt, Z_b, Z_g,b}}[\memconf],\s]{\MP_i} > 0$
    for all $\s$ and $\memconf \in [i_\s,b] \x \{0,1\}$.
  \end{lemma}
  \begin{proof}
  By fixing the finite-memory strategy $\optzstrat_{\mathtt{alt, Z_b,Z_g,b}}$,
  we obtain a finite Markov chain. Consider any BSCC in this Markov chain.
  In this BSCC, except for a null set of runs, either no Bailouts happen or infinitely many.
  In the former case, this BSCC behaves like playing $\optzstrat_{\Gain}$
  forever, which attains a strictly positive mean payoff in all dimensions almost
  surely, and thus a strictly positive expected mean payoff in each dimension $i$.
  In the second case, almost surely there happen infinitely many Bailouts,
  each starting at an every level $\ge Z_b$.
  Then, by the finiteness of the BSCC, we obtain that
  $\expectation[][]{T^{\Gain}} < \infty$.
  Moreover, by the definition of $\optzstrat_{\Bailout}$, the expected
  duration of the Bailout-phase is always finite, i.e.,
  $\expectation[][]{T^{\Bailout}} < \infty$.
  Thus, by linearity of expectations, 
  $\expectation[\mdp^\ast][\optzstrat_{\mathtt{alt, Z_b, Z_g,b}},\s]{\MP_i}
  \ge (v^1_i + v^2_i)/(\expectation[][]{T^{\Gain}} +
  \expectation[][]{T^{\Bailout}}) > 0$.
\end{proof}

    % \begin{align*}
    %     \expectation[][\mc]{\MP_i} &\geq \frac{\mu_i \cdot \min\lrc{\expectation[][]{T^{\Gain}}} - R \cdot \max\lrc{\expectation[][]{T^{\Bailout}}}}{\min\lrc{\expectation[][]{T^{\Gain}} + \max\lrc{\expectation[][]{T^{\Bailout}}}}}\\ 
    %     &\geq \frac{\mu_i \cdot \lrc{1-2\cdot d_i^k} \cdot \lrc{\lrc{k+1} \cdot \lrc{\frac{1}{\delta}-1} + \ceil{ \lrc{\log_{c_i}(\delta \cdot (1-c_i))}} } - R \cdot \lrc{{\size{\states^\ast} + \frac{2 \cdot d_\Bailout^{\size{\states^\ast}}}{1-d_\Bailout} + \frac{Z+h_\Bailout+R}{\mu}}} }{\lrc{1-2\cdot d_i^k} \cdot \lrc{\lrc{k+1} \cdot \lrc{\frac{1}{\delta}-1} + \ceil{ \lrc{\log_{c_i}(\delta \cdot (1-c_i))}}} + {\size{\states^\ast} + \frac{2 \cdot d_\Bailout^{\size{\states^\ast}}}{1-d_\Bailout} + \frac{Z+h_\Bailout+R}{\mu}} }
    % \end{align*}

The following technical \Cref{lem:lower_bounds_expected_sums_gain_bailout} (proof in \Cref{app:bounds}) shows that the
constants $v^1_i,v^2_i$ from \Cref{lem:sufficiency_condition_pos_mp} exist.
Recall that the finite Markov chains   
$\mc^{\Gain}$ and $\mc^{\Bailout}$ are obtained by fixing the memoryless
strategies $\optzstrat_{\Gain}$ and $\optzstrat_{\Bailout}$ in $\mdp^\ast$,
respectively.
Let $x_{\min,1}$ and $x_{\min,2}$ denote the minimal occurring
non-zero probabilities in these two Markov chains, respectively.
(They come from solutions of linear programs and can be chosen as only
exponentially small, i.e., described by a polynomial number of bits; cf.~\Cref{app:bounds}).
The proof works by applying general results about expected first passage times in
truncated Markov chains 
to the induced Markov chains $\mc^{\Gain}$ and $\mc^{\Bailout}$.
The general idea is that in the Gain-phase one has a general up drift in all
dimensions, and in particular in the first (energy) dimension.
It is thus unlikely to go down very far in the energy dimension,
even if the energy is truncated at $b$.
Thus, for a sufficiently large truncation point $b$ (actually $b=Z_g+1$ suffices),
the expected time spent in the Gain-phase is very large relative to the
expected time spent in the Bailout phase.
More exactly, the former increases exponentially in $b$, while the latter
is polynomial in $b$.
For a sufficiently large $b$ (exponential in $|\mdp^\ast|$),
the condition $v^1_i + v^2_i > 0$ is met.

\begin{restatable}{lemma}{constantsvonevtwo}\label{lem:lower_bounds_expected_sums_gain_bailout}
    Let $\mu_i > 0$ denote the lower bound on the mean payoff in dimension% R#3: l.460
    $i$ in any BSCC in the Markov chain $\mc^{\Gain}$ with corresponding
    computable constants $c_i,\,g_{\Gain},\,h_{\Gain}$
    $\lrc{\eqref{def:c},\eqref{def:g},\eqref{def:h}}$,
    and let $\mu$ denote the
    lower bound on the mean payoff in the $1^{st}$ dimension in any BSCC of
    $\mc^{\Bailout}$ with the corresponding constants
    $g_{\Bailout},h_{\Bailout}$.
    All the above constants, except $c_i$,
    can be chosen as
    at most exponential in $\size{\mdp^\ast}$ and
    $1/(1-c_i) \in \Ocompl\lrc{\exp\lrc{\exp\lrc{\size{\mdp^\ast}^{\Ocompl\lrc{1}}}}}$.
    
    Then there are constants $0 < C_1 < 1$, $C_2 > 0$, $C_3 > 0$, $C_4 > 0$,
    $C_5 > 0$, all exponential in $\size{\mdp^\ast}$
    and dependent only on $\mdp$,
    such that for  $k \eqdef \frac{2 \cdot
      \size{\states^\ast}}{x_{\min,1}^{\size{\states^\ast}}} \in
    \Ocompl\lrc{\exp\lrc{\size{\mdp^\ast}^{\Ocompl\lrc{1}}}}$, any
    $\delta \in (0,1)$ sufficiently small such that $\lrc{\size{\states^\ast}+1} \cdot \lrc{\frac{1}{\delta}-1} + \ceil{ \lrc{\log_{c_i}(\delta \cdot (1-c_i))}} \geq \frac{h_\Gain}{\mu_i}$ for all $2 \leq i \leq d$, one can choose \\
    $Z_g \eqdef Z_b\, +\, R\, +\, k \cdot R\, +\, \max_{i}\lrc{R \cdot \ceil{\log_{c_i}\lrc{\delta \cdot \lrc{1-c_i}}} - R + 1, h_\Gain}\, \in \Ocompl\lrc{e^{\size{\mdp^\ast}^{\Ocompl\lrc{1}}}\cdot \log\lrc{1/\delta}}$ and $b \eqdef Z_g + 1$ so that
    \begin{align*}
        \expectation[\mdp^\ast][\optzstrat_{\mathtt{alt, Z_b, Z_g, b}}[(Z_g,0)],q]{\lrc{Y_{T^{\Gain}}}_i} &\geq 
         C_1 \cdot \frac{1}{\delta} - C_2 \log_2 \lrc{\frac{1}{\delta}} - C_3
      & \eqdef v^1_i \\
        \expectation[\mdp^\ast][\optzstrat_{\mathtt{alt, Z_b, Z_g,
      b}}[(Z_b,1)],q]{\lrc{Y_{T^{\Bailout}}}_i} &\geq - C_4 \log_2
                                                  \lrc{\frac{1}{\delta}} - C_5
      & \eqdef v^2_i
    \end{align*}
    In particular, in order to satisfy the condition $v^1_i + v^2_i > 0$,
    it suffices to choose
    $1/\delta \in \Ocompl\lrc{\max(1/C_1,\max_{2\le j \le 5} C_j)^{\Ocompl\lrc{1}}}$.
    Since the constants $C_j$ are exponential in $\size{\mdp^{\ast}}$,
    and by the conditions on the other constants above, 
    the value $Z_g$, and hence the overall bound $b=Z_g+1$, can be chosen such that
    $b \in \Ocompl\lrc{\exp\lrc{\size{\mdp^\ast}^{\Ocompl\lrc{1}}}}$.
\end{restatable}

Now we can prove the first item of our main result.

\begin{proof}[Proof of \Cref{thm:main}(\Cref{res:inf_str ==> fin_str})]
Towards a contradiction, we assume that there exists
a state $\s^{\dagger}$ such that 
there is no finite-memory almost surely winning strategy from $\s^{\dagger}$
for $\obj\lrc{i_{\s^\dagger}}$ in the MDP $\mdp$.

First we consider the MDP $\mdp^\ast$.
The finite-memory strategy $\optzstrat_{\mathtt{alt, Z_b,Z_g,b}}[(i_{\s^\dagger},1)]$ from $\s^{\dagger}$
is energy-safe by construction and satisfies $\en_1(i_{\s^\dagger})$ surely.
Now consider the finite Markov chain induced by fixing this
finite-memory strategy in $\mdp^\ast$.
By \Cref{lem:sufficiency_condition_pos_mp} and \Cref{lem:lower_bounds_expected_sums_gain_bailout},
for a sufficiently large (exponential) $b$
it yields a strictly positive expected mean payoff
$v^1_i + v^2_i >0$ in every dimension $i \in [2,d]$ in every BSCC of this
Markov chain.
Since the Markov chain is finite, this implies that the mean payoff in every
dimension $i \in [2,d]$ is strictly positive almost surely.
Hence,
$\Prob[\mdp^\ast][\optzstrat_{\mathtt{alt,Z_b,Z_g,b}}[(i_{\s^\dagger},1)],\s^\dagger]{\MP_{[2,d]}\lrc{>0}}=1$
and thus
$\Prob[\mdp^\ast][\optzstrat_{\mathtt{alt,Z_b,Z_g,b}}[(i_{\s^\dagger},1)],\s^\dagger]{\obj\lrc{i_{\s^\dagger}}}=1$.
So there exists an almost surely winning finite-memory strategy from
$\s^\dagger$ for $\obj\lrc{i_{\s^\dagger}}$ in $\mdp^\ast$.
However, \Cref{lem:mdp_equiv_mdp^ast} then implies that there also exists
an almost surely winning finite-memory strategy from
$\s^\dagger$ for $\obj\lrc{i_{\s^\dagger}}$ in $\mdp$. Contradiction.
\end{proof}

\begin{remark}\label{rem:infix}
If $\optzstrat_{\mathtt{alt, Z_b,Z_g,b}}$ satisfies $\obj\lrc{i_\s}$ almost surely
from some state $\s$ then it also satisfies the stronger objective
$\obj\lrc{i_\s} \cap \infix\lrc{b}$ almost surely.
Consider a winning run induced by $\optzstrat_{\mathtt{alt, Z_b,Z_g,b}}$.
While the true energy might sometimes be higher than $b$, the energy remembered by
$\optzstrat_{\mathtt{alt, Z_b,Z_g,b}}$ is always $\le b$.
Even with this conservative under-approximation of the energy,
the run still satisfies the energy objective.
Therefore, in any winning run induced by $\optzstrat_{\mathtt{alt, Z_b,Z_g,b}}$,
the energy can never \emph{decrease} by more than $b$. Thus,
also $\infix\lrc{b}$ is satisfied almost surely.
\end{remark}

\section{Proof of \texorpdfstring{\Cref{res:mem_bound}}{memory bound}}\label{sec:mem_bounds}

Given some state $\s$, let $\zstrat = \memstrattuple$ be a finite-memory
strategy that is almost surely winning for $\obj\lrc{i_\s}$ 
(which exists by \Cref{res:inf_str ==> fin_str}).
We show there exists an almost surely winning strategy
$\zstrat'$ for $\obj\lrc{i_\s}$ such that the energy fluctuations are bounded by some 
constant which is exponential in $\size{\mdp}$.

First, inside any BSCC $B$ of $\mdp^\zstrat$, we construct an
almost surely winning strategy $\zstrat_B$ and upper bound the minimal safe energy levels and energy fluctuation 
while following $\zstrat_B$.
Using this, we upper bound the energy fluctuations in paths before reaching a BSCC. 
We use the fact that the set of states and transitions that occur in any BSCC
of a Markov chain induced by fixing some finite-memory strategy in an MDP
is an end component of this MDP (\cite[Theorem 3.2]{D-A1997}).

\begin{restatable}{lemma}{bsccmemboundTthree}\label{lem:BSCC_mem_bound}
    Let $B$ be a BSCC of $\mdp^\zstrat$ and let $\mdp\lrc{B}$ be the corresponding end component in $\mdp$ with states $\states_B$ and transitions 
    $\transition_B$. Then there is a strategy $\zstrat_B$, a bound 
    $b_B \in \Ocompl\lrc{\exp\lrc{\size{\mdp\lrc{B}}^{\Ocompl\lrc{1}}}}$ such that for any state $q \, \in \, \states_B$, there is a minimal safe energy level $j_q \eqdef i^{\mdp\lrc{B}}_q \leq 3 \cdot \size{\states_B} \cdot R$ such that 
    $ \Prob[\mdp\lrc{B}][\zstrat_B,q]{\obj\lrc{j_q} \, \cap\, \infix\lrc{b_B}} = 1$.
\end{restatable}
\begin{proof}[Proof Sketch. (Full proof in \Cref{app:mem_bounds}.)]
The idea is that for $\mdp\lrc{B}$ there are two cases.
In the first case it behaves similar to $\mdp^{\ast}$ from 
\Cref{sec:inf_str ==> fin_str}, in the sense that
it is possible to win $\Gain$ and $\Bailout$ almost surely, and thus
Energy-MeanPayoff can be won almost surely by switching between the
two strategies for $\Gain$ and $\Bailout$ like in the strategy $\optzstrat_{\mathtt{alt, Z_b,Z_g}}$.
Then one can invoke \Cref{lem:lower_bounds_expected_sums_gain_bailout,rem:infix} 
on $\mdp\lrc{B}$ to get an exponential bound $b_B$ such that $\Prob[\mdp\lrc{B}][\zstrat_B,q]{\obj\lrc{j_q} \, \cap\, \infix\lrc{b_B}} = 1$. 

If the first case does not hold then
$\mdp\lrc{B}$ is very restrictive,
and one can show that the energy level fluctuations are bounded by a constant
in $\Ocompl\lrc{\size{\states_B} \cdot R}$. 
\end{proof}

Since the minimal safe energy levels inside these end components are not too large, 
one can then bound the energy fluctuations in paths before they reach any such
end component $\mdp\lrc{B}$.
The following lemma is shown in \Cref{app:mem_bounds}.

\begin{restatable}{lemma}{transientmemboundTthree}\label{lem:Transient_mem_bound}
    Let $T$ denote the union of all $\states_B$ of every BSCC $B$ of
    $\mdp^\zstrat$,
    as in \Cref{lem:BSCC_mem_bound}.
    Then one can almost surely reach any state in $T$ with the corresponding minimal
    safe energy level with energy fluctuations of at most $5 \cdot \size{\states} \cdot R$.
\end{restatable}

\begin{proof}[Proof of \Cref{thm:main}(\Cref{res:mem_bound})]
  By \Cref{lem:BSCC_mem_bound,lem:Transient_mem_bound},
  for each state $\s$, one can choose a strategy $\zstrat$ and some constant 
  $b \in \Ocompl\lrc{\exp\lrc{\size{\mdp}^{\Ocompl\lrc{1}}}}$
  such that $\Prob[\mdp][\zstrat,\s]{\obj\lrc{i_\s}\, \cap \, \infix\lrc{b}} = 1$.
    This means if one encodes the energy levels between $[0,b]$ into the state space by discarding any excess energy above $b$ and redirecting all the 
    transitions which result in a negative energy to a losing sink (for
    $\vec{\MP}_{[2,d]}\lrc{> 0}$)
    and constructs this larger MDP $\mdp[0,b]$, 
    then there is a strategy $\zstrat'$ such that
    $\Prob[\mdp[0,b]][\zstrat',\tuple{\s,k}]{\vec{\MP}_{[2,d]}\lrc{> 0}} =1$
    for every $k \in [i_\s,b]$. 
    Then, by \Cref{lem:mr_gain_strategy},
    there also exists a memoryless (MR) strategy  
    $\optzstrat$ in $\mdp[0,b]$ which is almost surely winning 
    $\vec{\MP}_{[2,d]}\lrc{> 0}$ from $\tuple{\s,k}$.
    
    We can carry the memoryless strategy $\optzstrat$ in $\mdp[0,b]$ back to $\mdp$
    as a finite-memory strategy $\optzstrat_\mdp$ with memory $[0,b]$.
    It stores the
    encoded under-approximated energy level from $\mdp[0,b]$ in its finite memory instead. 
    Thus $\optzstrat_\mdp$ is a finite-memory strategy from $\s$ that
    satisfies $\obj\lrc{i_\s}$
    almost surely, and the size of its memory is bounded by
    $b \in \Ocompl\lrc{\exp\lrc{\size{\mdp}^{\Ocompl\lrc{1}}}}$.

    The strategy $\optzstrat_\mdp$ uses randomization, because $\optzstrat$
    from \Cref{lem:mr_gain_strategy} is MR.
    However, the MR strategy $\optzstrat$ for the mean payoff objective could
    be replaced by a deterministic strategy with an exponential number of
    memory modes. Hence the overall number of memory modes in the obtained
    deterministic version of $\optzstrat_\mdp$ is still only exponential.
  \end{proof}

\section{The Lower Bound (Proof of \Cref{res:no_fixed_mem_suffices})}\label{sec:memory-lowerbound}

In the previous sections we have shown that finite memory suffices
for almost surely winning strategies for the Energy-MeanPayoff objective.
However, the required memory depends on the given MDP.
We show that no fixed finite amount of memory is sufficient for all MDPs.
In fact, the required memory is exponential in the transition probabilities
even for an otherwise fixed 5-state MDP with just one controlled state, $R=1$ and $d=2$.

\begin{definition}\label{def:lowerbound}
Let $1 > \delta > 0$ and $\mdp_\delta=\mdptuple$ be an MDP with 2-dimensional rewards.
It has just one controlled state $\s$
with transitions $\s \to \s_l$ and $\s \to \s_r$.
From $\s_l$ there are two transitions $e_1 = (\s_l \to \s_l^1)$
and $e_2 = (\s_l \to \s_l^2)$.
Let $\probm(e_1) = (1 + \delta)/2$ and $\probm(e_2) = (1 - \delta)/2$
and $\vec{r}(e_1) = (+1,+1)$ and $\vec{r}(e_2) = (-1,-1)$.
$\s_l^1$ and $\s_l^2$ are random states which each have just one transition
back to $\s$ with probability $1$ and reward $\vec{0}$.
From $s_r$ there is only one transition $e_3$ back to $s$ with
probability $1$ and $\vec{r}(e_3) = (+1,-1)$.
\end{definition}

The following lemma directly implies the exponential lower bound
on the number of memory modes in \Cref{thm:main}(\Cref{res:no_fixed_mem_suffices}).

\begin{lemma}\label{thm:lowerbound}
  Consider the Energy-MeanPayoff objective.
  For every finite bound $m \in \N$ on the number of memory modes
  there exists a $\delta \eqdef 1/(6m) >0$
  such that the finite MDP $\mdp_\delta = \mdptuple$ from \Cref{def:lowerbound}
  satisfies the following properties.
  \begin{enumerate}
  \item
    $\exists \sigma'\ \probm^{\mdp_\delta}_{\sigma',\s}(\en_1(0)\,\cap\,\vec{\MP_2}\lrc{> 0})=1$,
    i.e., it is possible to win almost surely from $s$ in $\mdp_\delta$,
    even with initial energy $0$.
  \item
    For every finite-memory strategy $\sigma$ with $\le m$ memory modes we have
    $\probm^{\mdp_\delta}_{\sigma,\s}(\en_1(k)\,\cap\,\vec{\MP_2}\lrc{> 0})=0$
    for every $k \in \N$, i.e., $\sigma$ attains nothing in $\mdp_\delta$,
    regardless of the initial energy $k$.
  \item
    For $\mdp_\delta$ we have $|S|=5$, $d=1$ and $R=1$. 
    The number of memory modes required for an almost-surely winning strategy
in $\mdp_\delta$ is exponential in $|P|$ (and in $|\mdp_\delta|$).
  \end{enumerate}
\end{lemma}
\begin{proof}
Towards item 1, consider a strategy $\sigma'$ that plays as follows. It keeps
a counter that records the current energy, which is initially $0$.
Whenever the current energy is $0$, it plays $s \to s_r$, otherwise
it plays $s \to s_l$. Thus $\sigma'$ satisfies $\en_1(0)$ surely from $s$.
Since $\delta>0$ it follows from the classic Gambler's ruin problem
(with strictly positive expected gain, here in the first reward dimension) 
that $\sigma'$
plays $s \to s_r$ only finitely often, except in a nullset of the runs.% R#3: l.568
Therefore, the expected mean payoff (in the second dimension) under $\sigma'$ is
$(1+\delta)/2 - (1-\delta)/2 = \delta >0$.
Hence $\probm^{\mdp_\delta}_{\sigma',\s}(\vec{\MP_2}\lrc{> 0})=1$. Since the
energy objective is satisfied surely, we obtain
$\probm^{\mdp_\delta}_{\sigma',\s}(\en_1(0)\,\cap\,\vec{\MP_2}\lrc{> 0})=1$.

Towards item 2, let $\delta \eqdef 1/(6m) >0$ and
let $\sigma$ be a finite-memory strategy with $\le m$ memory modes.
Consider the finite-state Markov chain $\cal C$ that is induced by
playing $\sigma$ from $s$ in $\mdp_\delta$.
This Markov chain has $\le 5m$ states, since $\mdp$
has $5$ states and $\sigma$ has $\le m$ memory modes.
Let $B$ be any BSCC of $\cal C$ that is reachable from $s$ and the initial
memory mode of $\sigma$. In particular, $|B| \le 5m$.
In $B$ there must not exist any loop that does not contain $s_r$, because
otherwise the energy objective cannot be satisfied almost surely.
Thus every path in $B$ of length $\ge 5m$ must contain $s_r$ (and hence a reward
$(+1,-1)$) at least once. Therefore, the expected mean payoff in $B$ (in the
second reward dimension) is $\le 5m\delta - 1 = -1/6 < 0$.
Since this holds in every reachable BSCC, we obtain
$\probm^{\mdp_\delta}_{\sigma,\s}(\vec{\MP_2}\lrc{> 0})=0$ and thus
$\probm^{\mdp_\delta}_{\sigma,\s}(\en_1(k)\,\cap\,\vec{\MP_2}\lrc{> 0})=0$.

Towards item 3, the size of $\mdp_\delta$ follows from \Cref{def:lowerbound}.
By items 1 and 2,
the required number of memory modes
$m$ for an almost-surely winning strategy 
satisfies $m > 1/(6\delta)$.
Since $|P| = \Theta(\log(1/\delta))$
and $|\mdp_\delta| = \Theta(|P|)$, we obtain
$m = \Omega(\exp(|P|))$ and $m = \Omega(\exp(|\mdp_\delta|))$.
\end{proof}

The exponential lower bound on the required memory
does not require probabilities encoded in binary like in
\Cref{thm:lowerbound}.
One can construct an equivalent example with polynomially
many states where all transition probabilities are $1/2$.
This is because one can encode exponentially small probabilities $2^{-k}$
with a chain of $k$ extra states and transition probabilities $1/2$.

\section{Computational Complexity}\label{sec:conclusion}

We have shown that the existence of an almost surely winning strategy for the
Energy-MeanPayoff objective for a given state and initial energy level in an MDP
implies the existence of a deterministic such strategy with exponentially many
memory modes (unlike for Energy-Parity which 
requires infinite memory in general \cite{MSTW2017}).

A related problem is the decidability of the question whether a given state
in an MDP and a given initial energy level admit an almost surely winning
strategy for Energy-MeanPayoff.
This problem is decidable in \emph{pseudo-polynomial} time, using an algorithm very
similar to the one for Energy-Parity presented in \cite{MSTW2017}.
I.e., the time is polynomial, provided that the bound $R$ on the rewards
is given in unary. Transition probabilities in the MDP can still be represented
in binary.
The crucial point is that it suffices to witness the mere \emph{existence} of
an almost surely winning strategy, regardless of its memory.
Basically, it suffices that the algorithm proves that the infinite-memory
strategy $\optzstrat_{\mathtt{alt,Z_b,Z_g}}$ wins almost surely
(plus a small extra argument about a corner case where the energy fluctuates only
in a bounded region). The algorithm
does not need to compute the bound $b$ or to explicitly construct
the finite-memory strategy $\optzstrat_{\mathtt{alt,Z_b,Z_g,b}}$.

\begin{proposition}\label{prop:alg}
Let $\mdp = \mdptuple$ be an MDP with $d$-dimensional rewards on the edges
$\vec{r}: \transition \to [-R,R]^d$.
For any state $\s$ and $k \in \N$, the existence of an almost surely winning
strategy from $\s$ for the multidimensional Energy-MeanPayoff objective
$\en_1(k)\,\cap\,\vec{\MP_{[2,d]}\lrc{> 0}}$ is decidable in pseudo-polynomial
time (i.e., polynomial for $R$ in unary).
\end{proposition}
\begin{proof}
The proof is similar to the one for Energy-Parity presented in \cite{MSTW2017}.
We outline the differences below.
First, in the corner case where it is impossible to pump the energy up arbitrarily high almost surely
from some state $q$, the only possible way to win Energy-MeanPayoff
(resp.\ Energy-Parity) almost surely (if at all) is by using a non-nullset of
runs where the energy only ever fluctuates in a bounded region.
In that case, the size of the energy fluctuations in these runs can safely be restricted
to a region that is polynomial in $|S|\cdot R$, and thus pseudo-polynomial in
$|\mdp|$ \cite{MSTW2017}.
It thus suffices to win multi-dimensional $\MP_{[2,d]}\lrc{> 0}$ almost surely in a
derived MDP $\mdp'$ where the bounded energy is encoded into the states.
Deciding this requires time polynomial in $|\mdp'|$ \cite{CD2011,Gimbert2011ComputingOS} and thus
pseudo-polynomial in $|\mdp|$.
The winning situations of the corner case can then be encoded into $\mdp$,
yielding a derived MDP $\mdp'$ of pseudo-polynomial size, where
Energy-MeanPayoff can be won almost surely if and only if it can be won almost
surely by a combination of $\Gain$ and $\Bailout$ strategies, i.e., by strategy
$\optzstrat_{\mathtt{alt,Z_b,Z_g}}$.
Therefore it suffices to compute the states (and minimal initial energy levels $k$)
where $\Gain$ and $\Bailout(k)$ can be won almost surely.
The objective $\Bailout(k) \eqdef \en_1(k) \, \cap\, \MP_1\lrc{> 0}$ is
exactly the same as the $\Bailout$ objective analyzed in \cite{MSTW2017},
and winning it almost surely is decidable in pseudo-polynomial time.
Our objective $\Gain \eqdef \MP_{[1,d]}\lrc{> 0}$ differs from the
$\Gain$ objective considered in \cite{MSTW2017} (which was
$\MP_1\lrc{> 0}\,\cap\, \parity$), but winning it almost surely is
still decidable in polynomial time \cite{CD2011,Gimbert2011ComputingOS}
by solving a linear program.
So overall the algorithm runs in pseudo-polynomial time.
\end{proof}

\newpage
\bibliography{conferences,journals,ref1}

% Appendix
\newpage
\appendix
\section{Bounds in Markov Chains}\label{app:mc}

This section shows some generic results for Markov chains with transition
rewards. We show bounds on the expected arrival time (aka first passage time)
of situations when the total reward reaches particular levels,
under the condition that the total reward is truncated to remain inside
some interval $[a,b]$. E.g., the total reward might hit the upper limit $b$
many times (and be truncated there)
before arriving at the lower level $a$ for the first time.
These bounds are later applied to Markov chains obtained by fixing certain
finite-memory strategies in MDPs, and they are used in the proof of
\Cref{lem:lower_bounds_expected_sums_gain_bailout}.

Let $\mc$ be a strongly connected Markov chain with state space $\states$, one-step transition probability matrix $\probp$ and stationary distribution $\pi > \vec{0}$.
Consider a reward function on the edges
$r: \transition \to [-R,R]$ (alternately it can be seen as a vector in $[-R,R]^{\transition}$)% R#3: l.730
such that the average reward gained in the limit is positive, i.e.,
$\mu \eqdef \sum_{e = (\s,\s') \in \transition} f_e \cdot r(e) >0$ where $f_e \eqdef \pi(\s) \cdot \probp(\s)(\s')$ denotes the long term relative 
frequency of the edge $e$. We begin by defining some functions on $\states^{\om}$. Let $\rho = q_0q_1 \ldots $ be a generic infinite word. 

Recall that $X_n$ denotes the state of a Markov chain at time $n$ and $Y_n$ is the sum of the rewards until time $n$.

Let $x_{\min}$ denote the minimum occurring probability in $\mc$. The size of the Markov chain with the reward structure $\size{\mc}$ is defined as the 
total number of bits required to represent each state, edge, probability and reward in binary. We assume that all the probabilities are rational and 
rewards integers. Let 
\begin{equation}\label{def:h}
    h \eqdef \frac{2 \size{\states} R}{x_{\min}^{\size{\states}}}
\end{equation}

\begin{lemma}\label{lem:potential_vector}~\cite[Theorem 3.4]{BKK2014}
    Let $u_{\s} \eqdef \sum_{\s' \in \successors{\s}} \probp(\s)\lrc{\s'} \cdot r(\tuple{\s,\s'})$ be the expected reward gained after taking an edge
    from state $\s$. There exists $\nu \in [0,h]^{\states}$  such that
    \[
        u + \probp \nu = \nu + \vec{1}\mu
    \]
    where $\vec{1}$ on the RHS is a vector of all $1$'s.
\end{lemma}

\begin{fact}\label{fact:martingale}
    Let $\nu$ be the vector from~\Cref{lem:potential_vector} and $\s$ be the
    start state, i.e., $X_0=\s$. Then the sequence of random variables given by
    \[
        M_n^{\s} \eqdef Y_n + \nu(X_n) - n\mu
    \]
    is a martingale for all $\s$.
\end{fact}

Since the average mean payoff $\mu >0$ is strictly positive,
$\liminf_{n \tendsto \infty} Y_n = \infty$ almost surely.
In the above setting, suppose that we bound the total reward gained to lie in
some interval $[a,b]$ for some integers $a < 0 < b$, i.e., we
define a new sequence of functions inductively as follows.
\begin{align*}
    Y^{[a,b]}_0 &\eqdef Y_0 = 0 \\
    Y^{[a,b]}_n &\eqdef \max(a,\min(b,Y^{[a,b]}_{n-1}+ r(\tuple{X_{n-1},X_n}))) \quad \mbox{for $n \geq 1$}% R#3: l.754
\end{align*}

Considering the sequence $Y^{[a,b]}_n$,
let $T^{[a,b]}_a$, $T^{[a,b]}_b$  be the functions which denote the first hitting time of the % R#3: l.756
left boundary $a$ and right boundary $b$ respectively.
Clearly, $Y^{[a,b]}_n \le Y_n$ before $T^{[a,b]}_a$, because the only possible
difference is that $Y^{[a,b]}_n$ loses something when hitting the right border
$b$.
One of the advantages of $Y^{[a,b]}_n$ is that it can be described using only
a finite number of bits for any $n$, because of its boundedness,
whereas this is not the case for $Y_n$.
This is useful in situations where such an under-approximation suffices instead
of remembering the exact reward gained.
However, the bounding of the reward also changes the behaviour of $Y_n$.
For example, the probability of $Y_n$ falling below $a$ infinitely often is
zero, i.e., 
$\Prob{\always \eventually (Y_n \leq a)} = 0$ for a positive mean payoff $\mu >0$.
The same does not generally hold for $Y^{[a,b]}_n$, which might fall below $a$
infinitely often almost surely.

We want to show that (on average) $Y^{[a,b]}_n$ hits the lower bound $a$ much
less frequently than the upper bound $b$.
That is, we derive a lower bound on the expected time it takes to hit the
lower bound $a$, and an upper bound on the expected time it takes to hit the % R#3: l.768
upper bound $b$.

\begin{remark}
Although we denote the functions by $T^{[a,b]}_a$, $T^{[a,b]}_b$ and $Y^{[a,b]}_n$ etc., note that there is an implicit assumption that the initial sum is $0$ which lies between $a$ and $b$. Therefore, the hitting times are actually parametrized by three numbers $a$, $b$ and $x$ such that $a < x < b$ given by $T^{[a,x,b]}_a$, $T^{[a,x,b]}_b$, $Y^{[a,x,b]}_n$. But we continue to represent with just the boundary points as all the functions are invariant under translation \ie, $T^{[a,x,b]}_a = T^{[a',b']}_{a'}$- where $a' = a-x$ and $b' = b-x$. Moreover, for $a < x_1 < x_2 < b$
$$ T^{[a,x_1,b]}_a \leq T^{[a,x_2,b]}_a $$
Or in other form $T^{[a-x_1,b-x_1]}_{a-x_1} \leq T^{[a-x_2,b-x_2]}_{a-x_2}$ for all $a < x_1 < x_2 < b$.% R#3: l.770
\end{remark}

\subsection{Upper bound on the hitting time in the direction of drift}

Given an initial state $\s$, let the random variable $T_b$ denote the first
time $Y_n$ is $\geq b >0$.
\begin{equation}\label{eq:hitting_time_b}
    T_b \eqdef \inf \setcomp{n}{Y_n \geq b}
\end{equation}
Since the overall drift is in the positive direction, i.e.,
$\mu > 0$, it is immediate that the expectation of $T_b$ is finite for every $b$. Also, the event ${T_b = n}$ can be determined by looking at the first $n$ steps of any run, so $T_b$ is a stopping time w.r.t the natural filtration of the Markov chain.% R#3: l.776

\begin{fact}
    For all states $\s$ and $b > 0$, $T_b$ is a stopping time and $\expectation[][\s]{T_b} < \infty$.
\end{fact}

Now we show some bounds on this expected stopping time.

\begin{lemma}\label{lem:bound_hitting_time_b}
    For all states $\s$ and $b > 0$
    \[
        \frac{b-h}{\mu} \leq \expectation[][\s]{T_b} \leq \frac{b+h+R}{\mu}
      \]
      where $h$ is the constant from \Cref{def:h}.
\end{lemma}
\begin{proof}
    Since $\expectation[][\s]{T_b} < \infty$ and the martingale $M_n^{\s}$ has bounded step size, %($\abs{m_{n+1}^{\s} - m_n^{\s}} \leq 2R+h$) 
    we can apply optional stopping theorem to get
    \[
        \expectation[][\s]{M_{T_b}^{\s}} = \expectation[][\s]{M_0^{\s}} = \nu(\s)
      \]
      Since $0 \le \nu(\s) \le h$, it follows that
    \begin{equation*}
        0 \leq \expectation[][\s]{Y_{T_b} + \nu(X_{T_b}) - T_b \mu} \leq h.
    \end{equation*}

    Simplifying by using linearity of expectation and the fact that $b \leq Y_{T_b} < b+R$ and $0 \leq \nu(X_{T_b}) \leq h$, we get
    \begin{align*}
        0 &\leq b+R+h-\expectation[][\s]{T_b}\mu \\
        b+0-\expectation[][\s]{T_b}\mu &\leq h
    \end{align*}
    Rearranging terms and noting that $\mu > 0$ yields the required bounds.
\end{proof}

\Cref{lem:bound_hitting_time_b} shows that, starting from any state,
the expected time to hit any upper total reward boundary $b >0$ is
asymptotically linear in $b$. To get the upper bound for $T^{[a,b]}_b$, observe that 
$\lrc{Y_n \leq Y^{[a,b]}_n} \given \lrc{T_b \geq n}$ implying $T^{[a,b]}_b \leq T_b$. Hence we obtain the following as a corollary.% R#3: l.796

\begin{corollary}\label{cor:upper_bound_expectation_hitting_time_right_bdry_bscc}
    $$\expectation[][\s]{T^{[a,b]}_b} \leq \frac{b+h+R}{\mu}$$
\end{corollary}

\subsection{Bounding the probability of going in the opposite direction of the drift}

We now derive an upper bound on the probability
that the total reward ever falls below some large negative number $a <0$.
This will be used later to derive the required inequalities for $T^{[a,b]}_a$.

Given $a <0$, let
the random variable $T_a$ denote the first time that the total reward is $\leq a$.
\begin{equation}\label{eq:hitting_time_a}
    T_a \eqdef \inf \setcomp{n}{Y_n \leq a}
\end{equation}

Since the average reward $\mu >0$ is positive,
it is unlikely that $Y_n$ ever hits large negative numbers. The following lemma quantifies this intuition.
Recall that $h = \frac{2 \size{\states} R}{x_{\min}^{\size{\states}}}$ and let
\begin{align}
    \eta &\eqdef \mu + h + R \label{def:eta} \\
    c &\eqdef e^{\frac{-\mu^2}{2\eta^2}} \label{def:c}
\end{align}

\begin{lemma}\label{lem:exp_decay}
  For all $a \leq -h$ and states $\s$ we have
  \[
  \Prob[\mc][\s]{T_a < \infty} \leq
  \frac{c^{\ceil*{\frac{\abs{a}}{R}}}}{1-c}.
  \]
\end{lemma}
\begin{proof}
  Observe that the consecutive terms of the sequence $M^{\s}_n$
  differ by at most $\eta$. Consider the event $T_a = n$.
  From our assumption $a \leq -h$ we obtain that $a+h\leq 0$,
  and thus
    \begin{align*}
        M^{\s}_{n} - M^{\s}_0 &= \lrc{Y_n - Y_0} + \lrc{\nu\lrc{X_n} - \nu\lrc{X_0}} - n\mu \\
        &\leq a + h - n\mu \\
        &\leq -n\mu
    \end{align*}
    Hence, using the Azuma-Hoeffding inequality, we obtain
    \[
        \Prob{T_a=n} \leq \Prob{M^{\s}_n - M^{\s}_0 \leq -n\mu} \leq e^{\frac{-n^2\mu^2}{2n \eta^2}} = \lrc{e^{\frac{-\mu^2}{2\eta^2}}}^n    
    \]
    Since we have defined $c \eqdef e^{\frac{-\mu^2}{2\eta^2}}$,
    we see that $\Prob{T_a=n} \leq c^n$.
    Since $T_a \ge \ceil*{\frac{\abs{a}}{R}}$,
    we get that
    \[
    \Prob{T_a < \infty} = \sum_{n=\ceil*{\frac{\abs{a}}{R}}}^{\infty}
    \Prob{T_a=n} \leq \frac{c^{\ceil*{\frac{\abs{a}}{R}}}}{1-c}.
    \]
\end{proof}

The above lemma provides a bound for the case of $Y_n$ where the total reward
is unrestricted. It is exponentially more unlikely that $Y_n$ ever drops as low
as $a$ when $a \tendsto -\infty$.

We need a lower bound on the expected time to hit the left bound $a$
for the bounded random variable $Y^{[a,b]}_n$
(since we are interested in a finite-memory strategy).
To do so, we first lower bound it by another variable which is simpler to analyse. Recall that $T^{[a,b]}_a = \inf \setcomp{n}{Y^{[a,b]}_n \leq a}$. Define a new sequence of random variables $Y^{\prime[a,b]}_n$ inductively as follows.
\begin{align*}
	Y^{\prime[a,b]}_0 &= Y^{[a,b]}_0 = 0 \\
	Y^{\prime[a,b]}_n &= \left \{ \begin{aligned}
		Y^{\prime[a,b]}_{n-1} + r(\tuple{X_{n-1},X_n}) \  &\textrm{if }\  a < Y^{\prime[a,b]}_{n-1} + r(\tuple{X_{n-1},X_n}) < b & \textrm{: rule} \\
		0  \  &\textrm{if } \  Y^{\prime[a,b]}_{n-1} + r(\tuple{X_{n-1},X_n}) \geq b & \textrm{: reset} \\
		a \  &\textrm{if } \  Y^{\prime[a,b]}_{n-1} + r(\tuple{X_{n-1},X_n}) \leq a & \textrm{: hit}
	\end{aligned}
	\right.
\end{align*}

Intuitively, the behaviour of 
$Y^{\prime[a,b]}_n$ is similar to that of
$Y^{[a,b]}_n$, except when it hits/exceeds $b$.
Instead of clamping to $b$, 
$Y^{\prime[a,b]}_n$ `resets' and behaves as if it is starting
from the current state $X_n$.
Let $T^{\prime[a,b]}_a$ denote the first time that
$Y^{\prime[a,b]}_n$ hits the left bound $a$, i.e.,
$T^{\prime[a,b]}_a \eqdef \inf \setcomp{n}{Y^{\prime[a,b]}_n \leq a}$.

\begin{claim}
For all $n\ge 0$ and $b > 0 > a$ we have $Y^{[a,b]}_n \geq Y^{\prime[a,b]}_n$.
Consequently, $T^{[a,b]}_a \geq T^{\prime[a,b]}_a$.
\end{claim}
\begin{proof}
By induction on $n$.

In the base case $n=0$ we have $Y^{[a,b]}_0 = Y^{\prime[a,b]}_0$.

For the induction step let $n >0$.

If $Y^{\prime[a,b]}_{n-1} + r(\tuple{X_{n-1},X_n}) \ge b$ then $Y^{[a,b]}_{n-1} + r(\tuple{X_{n-1},X_n}) \geq b$,
by induction hypothesis. 
Then $Y^{\prime[a,b]}_n = 0 < b = Y^{[a,b]}_n$.

Else if $Y^{\prime[a,b]}_{n-1} + r(\tuple{X_{n-1},X_n}) \leq a$, then $Y^{\prime[a,b]}_n = a \leq Y^{[a,b]}_n$.

Otherwise, we have $a < Y^{\prime[a,b]}_n) = Y^{\prime[a,b]}_{n-1} + r(\tuple{X_{n-1},X_n}) < b$
and by induction hypothesis
$Y^{\prime[a,b]}_n = Y^{\prime[a,b]}_{n-1} + r(\tuple{X_{n-1},X_n}) \le
\min(b,Y^{[a,b]}_{n-1} + r(\tuple{X_{n-1},X_n})) = Y^{[a,b]}_n$.
\end{proof}

\begin{comment}
Let $\mc^{[a,b]}$ denote the original Markov chain but with the overall energy bounded between $[a,b]$ for some $a \leq -(R+B) < 0 < R < b$ $\ie$ the state space of $\mc^{[a,b]}$ is $\states \x [a,b]$ with the second coordinate being the bounded sum of rewards. Whenever the energy increases beyond $b$ in $\mc$, it is clamped down to $b$ and also any decrease below $a$ is clamped to $a$.

To get a lower bound on the expected time to hit the left boundary (states of the form $\tuple{\s,a}$), we construct a similar Markov chain $\mc'^{[a,b]}$ which is slightly simpler to analyse. Specifically, any transition of the form $\tuple{\s,x} \energymove{p} \tuple{\s',b}$ is now replaced with $\tuple{\s,x} \energymove{p} \tuple{\s',r(\s')}$. Whenever we are about to hit the right boundary, we simply `reset' the energy level in our Markov chain and start again.

Remember that $T^{[a,b]}_a$ is the random variable which denotes the hitting time of left boundary in $\mc^{[a,b]}$. In a similar spirit, let 
$L'^{[a,b]}$ denote the hitting time in $\mc'^{[a,b]}$. It is immediate that $L'^{[a,b]}$ is a lower bound for $T^{[a,b]}_a$.
\end{comment}

To lower bound
$\expectation[][\s]{T^{\prime[a,b]}_a}$,
we split each run at the reset points (when $b$ is hit or exceeded)
and evaluate the expected number of resets that 
happen before hitting
$a$.
Formally, let $V^{\prime[a,b]}_a$ be the random variable denoting the number of resets before hitting $a$.
\[
V^{\prime[a,b]}_a \eqdef \sum_{i=1}^{T^{\prime[a,b]}_a} 1_{\lrc{Y^{\prime[a,b]}_{i-1} + r(\tuple{X_{i-1},X_i})\, \geq \, b}}.% R#3: l.850
\]
We analogously define the random variables $V^{[a,b]}_a$ and $V^b_a$ for the
non-primed random variable $Y^{[a,b]}_n$ and the unbounded random variable
$Y_n$, respectively.

For the argument to work, we have to space the bounds $a$ and $b$ sufficiently far % R#3: l.853
apart. In the rest of the section we assume that $a \leq -h < 0 < b$. 

% \begin{assumption}
%    $$a \leq -(R+B) < 0 < R < b$$
% \end{assumption}

Since the step size is bounded by $R$,
the constants $\alpha \eqdef \ceil{\frac{\abs{a}}{R}}$ and 
$\beta \eqdef \ceil{\frac{\abs{b}}{R}}$ are universal lower bounds on
the minimum time it takes to hit the left bound $a$ and the right bound $b$,
respectively.

\begin{align*}
	\expectation[][\s]{T^{\prime[a,b]}_a} &= \sum_{n=0}^{\infty} \Prob[][\s]{T^{\prime[a,b]}_a > n} \\
\end{align*}

Since hitting $a$ takes at least $\alpha > 0$ steps, the probability
$\Prob[][\s]{T^{\prime[a,b]}_a > n} = 1$
for all $0 \leq n \leq \alpha-1$. Thus, the summation can be simplified to
\begin{align*}
	&\alpha + \sum_{n=0}^{\infty} \Prob[][\s]{T^{\prime[a,b]}_a > n + \alpha} \\
	=& \alpha + \sum_{j=0}^{\infty} \sum_{k=0}^{\beta-1} \Prob[][\s]{T^{\prime[a,b]}_a > j \cdot \beta + k + \alpha} \\
	\geq& \alpha + \sum_{j=0}^{\infty} \beta \cdot \Prob[][\s]{T^{\prime[a,b]}_a > \lrc{j+1} \cdot \beta  + \alpha -1} \\
	=& \alpha + \sum_{j=0}^{\infty} \beta \cdot \Prob[][\s]{T^{\prime[a,b]}_a \geq \lrc{j+1} \cdot \beta  + \alpha} \\
	\geq& \alpha + \sum_{j=0}^{\infty} \beta \cdot \Prob[][\s]{T^{\prime[a,b]}_a \geq \lrc{j+1} \cdot \beta  + \alpha \,\wedge\, V^{\prime[a,b]}_a \geq j+1} \\
    \geq& \alpha + \sum_{j=0}^{\infty} \beta \cdot \Prob[][\s]{V^{\prime[a,b]}_a \geq j+1} \\
\end{align*}
where the last inequality is justified by the fact that resetting at least $j+1$ times implies that the time taken to hit
the left bound $a$ would be at least $\lrc{j+1}\cdot \beta + \alpha$.

\begin{claim}
Let $0 < \delta < 1$ and let $c$ be as in \Cref{def:c}.
By choosing $a \eqdef \min\lrc{-R \ceil{ \lrc{\log_c(\delta \cdot (1-c))}} + R-1, -h}$
we obtain
$\Prob[][\s]{V^{\prime[a,b]}_a = 0} \leq \delta$ for any start state $\s$.
Moreover, it holds that $\Prob[][\s]{V^{\prime[a,b]}_a \geq j+1} \geq \lrc{1-\delta}^{j+1}$.
\end{claim}
\begin{proof}
	From our choice of $a$ and \Cref{lem:exp_decay}, it follows that 
	$\Prob[\mc][\s]{T_a < \infty} \leq \frac{c^{\ceil*{\frac{\abs{a}}{R}}}}{1-c} \leq \delta$. 
	Consider the event $V^{\prime[a,b]}_a = 0$ when starting from $\s$. 
	This means any run in this event doesn't hit the reset transitions,
        which implies that the probability of this event doesn't change when 
	considering $Y_n$ or $Y^{[a,b]}_n$ instead of the sequence $Y^{\prime[a,b]}_n$.
	\[
		\Prob[][\s]{V^{\prime[a,b]}_a = 0} = \Prob[][\s]{V^{[a,b]} = 0} = \Prob[][\s]{V^b_a = 0}. 
	\]
	But the event $V^b_a = 0$ is exactly equivalent to the event $T_a < T_b$ in $\mc$ which further implies that 
	$T_a < \infty$ in $\mc$. Therefore, we have that
	\begin{equation} \label{eq:event_no_right_bdry_visits_upper_bound}
		\Prob[][\s]{V^{\prime[a,b]} = 0} = \Prob[][\s]{T_a < T_b} \leq \Prob[\mc][\s]{T_a < \infty} \leq \delta.
	\end{equation}
	We can then prove the required claim by induction on $j$.

        Base case $j=0$:
        $\Prob[][\s]{V^{\prime[a,b]}_a \geq 1} = \Prob[][\s]{V^{\prime[a,b]}_a > 0} = 1 - \Prob[][\s]{V^{\prime[a,b]}_a = 0} \geq 1-\delta$, by \Cref{eq:event_no_right_bdry_visits_upper_bound}.
    
    Induction step:
    \begin{align*}
        \Prob[][\s]{V^{\prime[a,b]}_a \geq j+2} &= \Prob[][\s]{V^{\prime[a,b]}_a \geq j+2, V^{\prime[a,b]}_a \geq j+1} \\
        &=  \Prob[][\s]{V^{\prime[a,b]}_a \geq j+2 \given V^{\prime[a,b]}_a \geq j+1} \cdot  \Prob[][\s]{V^{\prime[a,b]}_a \geq j+1} \\
    \end{align*}
    Let $T'^{[a,b]}_{b,j}$ denote the time taken for the $j^{th}$ visit to energy level $b$ in $\mc$. It is clear that $T'^{[a,b]}_{b,j}$ is a stopping time for every $j$. Using strong Markov property, one can simplify the above conditional probability to get the required result.
    \begin{align*}
        & \geq \lrc{1-\delta}^{j+1} \cdot \sum_{\s' \in \states} \Prob[][\s]{V^{\prime[a,b]}_a \geq j+2 \given V^{\prime[a,b]}_a \geq j+1, X_{T'^{[a,b]}_{b,j+1}} = \s'} \cdot \Prob[][\s]{X_{T'^{[a,b]}_{b,j+1}} = \s' \given V^{\prime[a,b]}_a \geq j+1} \\
        & = \lrc{1-\delta}^{j+1} \cdot \sum_{\s' \in \states} \Prob[][\s']{V^{\prime[a,b]}_a \geq 1} \cdot \Prob[][\s]{X_{T'^{[a,b]}_{b,j+1}} = \s' \given V^{\prime[a,b]}_a \geq j+1} \\
        & \geq \lrc{1-\delta}^{j+1} \cdot \sum_{\s' \in \states} \lrc{1-\delta} \cdot \Prob[][\s]{X_{T'^{[a,b]}_{b,j+1}} = \s' \given V^{\prime[a,b]}_a \geq j+1} \\
        &= \lrc{1-\delta}^{j+2}
    \end{align*}
\end{proof}

Thus, we get that
\begin{align}
	\expectation[][\s]{T^{[a,b]}_a} \geq \expectation[][\s]{T^{\prime[a,b]}_a} &\geq \alpha + \sum_{j=0}^{\infty} \beta \cdot \Prob[][\s]{V^{\prime[a,b]}_a \geq j+1} \nonumber \\
	&\geq \alpha + \sum_{j=0}^{\infty} \beta \cdot \lrc{1-\delta}^{j+1} \nonumber \\
	&= \alpha + \beta \cdot \frac{1-\delta}{\delta} \nonumber \\
	&= \beta \cdot \frac{1}{\delta} + \lrc{\alpha-\beta}. \label{eq:expectation_left_bdry_hitting_time_bscc}
\end{align}

We aim to make the interval $[a,b]$ as small as possible
(since later the size of the memory in our strategies will be proportional to $\size{b-a}$).
Due to the different influences of the parameters $a$ and $b$, it is % R#3: l.908. 
better to fix $b$ as $1$ and make $a$ smaller (more negative).
The $\beta$ will then be $1$ as well.

\begin{lemma}\label{lem:lower_bound_expectation_left_bdry_hitting_time_bscc}
    For $0 < \delta < 1$ and $a = \min\lrc{-R \ceil{ \lrc{\log_c(\delta \cdot
          (1-c))}} + R-1, -h}$ we have
    \[
        \expectation[][\s]{T^{[a,1]}_a} \geq \frac{1}{\delta} + \ceil{\lrc{\log_c(\delta \cdot (1-c))}} - 1.
    \]
\end{lemma}
\begin{proof}
    Since $b=1$ and $\abs{a} \geq R \lrc{\ceil{ \lrc{\log_c(\delta \cdot (1-c))}} - 1} +1$, $\beta = 1$ and $\alpha  = \ceil{\frac{\abs{a}}{R}} \geq \ceil{\lrc{\log_c(\delta \cdot (1-c))}}$. Substituting these values in \Cref{eq:expectation_left_bdry_hitting_time_bscc} gives the required bound.
\end{proof}

As $\delta \tendsto 0$, the expected hitting time grows as
$\approx \frac{1}{\delta} = \exp\lrc{\log\lrc{1/\delta}}$, i.e.,
exponentially in $\log\lrc{\frac{1}{\delta}}$, whereas the memory $\approx a$ required to achieve this lower bound would be proportional to $R\lrc{\frac{\log\lrc{1/\delta}}{\log\lrc{1/c}}} = \frac{2\eta^2R}{\mu^2}\log\lrc{1/\delta}$, linear in $\log(1/\delta)$.

\subsection{General Markov Chains}
To get a lower bound on $\expectation[][\s]{T^{[a,b]}_a}$ when $\mc$ is not
strongly connected, one has to account for the time spent in transient
states. Fortunately, the probability to spend a large amount of time outside a BSCC falls exponentially and this allows the analysis from the previous
subsections to carry over with a minimal increase in the size of the interval
$[a,b]$ required for general Markov chains. Assume $\AS\lrc{\MP\lrc{> 0}} =
\states$, i.e., every BSCC of $\mc$ has positive average mean payoff. Let $C \subseteq \states$ denote all the recurrent states. Compute the constants $\eta_{G},\, \mu_{G},\, c_{G},\, h_{G}$ from \Cref{lem:exp_decay} for each BSCC $G$ and let $\eta,\, \mu,\, c,\, h$ be the maximum over all BSCC's. Let $T_C$ denote the hitting time of some BSCC. For some positive integer $k$, let $Z_k$ denote the event  $T_C \leq k$. By the tower property, one has
$$ \expectation[][\s]{T^{[a,b]}_a} = \expectation[][\s]{\expectation{T^{[a,b]}_a \given 1_{Z_k}}} = \expectation[][\s]{T^{[a,b]}_a \given \complementof{Z_k}} \cdot \Prob[][\s]{\complementof{Z_k}} + \expectation[][\s]{T^{[a,b]}_a \given Z_k} \cdot \Prob[][\s]{Z_k}.$$

Since we are only interested in a lower bound, we can ignore the low probability event $\complementof{Z_k}$ as $T^{[a,b]}_a$ is a non-negative random variable.
\begin{equation}\label{eq:expectation_left_bdry_hitting_time_lower_bnd}
    \expectation[][\s]{T^{[a,b]}_a} \geq \expectation[][\s]{T^{[a,b]}_a \given Z_k} \cdot \Prob[][\s]{Z_k}. 
\end{equation}

If $x_{\min}$ is the minimum occurring probability in $\mc$, then let
\begin{equation} \label{def:g}
    g \eqdef
        \exp\lrc{\frac{-x_{\min}^{\size{\states}}}{\size{\states}}}
\end{equation}

\begin{lemma}\label{lem:lower_bound_probability_hitting_BSCC}
    Let $y_{\min}$ denote the minimum occurring probability in $\mc$ outside every BSCC. 
     Then there exists $0 \leq g < 1$ such that for all $k > \size{\states}$, 
    $$\Prob[][\s]{\complementof{Z_k}} \leq 2\cdot g^k.$$
\end{lemma}
\begin{proof}
    The proof is similar to ~\cite[Lemma 5.1]{BKK2014}. Assume $y_{\min} \neq 1$ (the other case is trivial $\Prob[][\s]{\complementof{Z_k}}=0$). 
    This implies $y_{\min} \leq \frac{1}{2}$.
    Let $n=\size{\states}$. From any state $\s$, there will be a path of length at most $n-1$ to a state in $C$, $\implies$ for all states $\s$, $\Prob[][\s]{T_C < n} \geq y_{\min}^{n-1} \geq y_{\min}^{n}$. Dividing the run into segments of length $n-1$, one gets
    \begin{align*}
        \Prob[][\s]{\complementof{Z_k}} &= \Prob[][\s]{T_C > k} \\
        &\leq \Prob[][\s]{T_C \geq k} \\
        &\leq \lrc{1-y_{\min}^n}^{\floor{\frac{k-1}{n-1}}} \\
        &\leq 2 \cdot \lrc{\exp\lrc{\frac{1}{n}\log\lrc{1-y_{\min}^n}}}^k \\
        &\leq 2 \cdot g^k \lrc{g = \exp\lrc{\frac{-y_{\min}^{n}}{n}}}
    \end{align*} 
\end{proof}

From the above lemma, we get a lower bound on $\Prob[][\s]{Z_k}$ for $k > n$. Before computing a lower bound on $\expectation[][\s]{T^{[a,b]}_a \given Z_k}$, we choose $a$ to be sufficiently negative so that it is never the case that $T^{[a,b]}_a \leq k$.

\begin{lemma}\label{lem:lower_bound_expectation_left_bdry_hitting_time_given_Z_k}
    For any $0 < \delta < 1$, choosing $a = \min\lrc{-R \ceil{ \lrc{\log_c(\delta \cdot (1-c))}} + R-1, -h} - k \cdot R$ and $b = 1$
    \[
        \expectation[][\s]{T^{[a,b]}_a \given Z_k} \geq \lrc{k+1} \cdot \lrc{\frac{1}{\delta}-1} + \ceil{ \lrc{\log_c(\delta \cdot (1-c))}}
    \]
\end{lemma}
\begin{proof}
    Since $T_C \leq k < T^{[a,b]}_a \given Z_k$, $Y^{[a,b]}_{T_C} \geq -k \cdot R$. 
    
    Let $I \eqdef \setcomp{\tuple{y,q}}{Y^{[a,b]}_{T_C}=y \wedge X_{T_C}=q \wedge \Prob[][\s]{Y^{[a,b]}_{T_C}=y, X_{T_C}=q \given Z_k} > 0}$
    \[ 
    \expectation[][\s]{T^{[a,b]}_a \given Z_k} = \sum_{(y,q) \in I} \expectation[][\s]{T^{[a,b]}_a \given Z_k, \,Y^{[a,b]}_{T_C}=y, \,X_{T_C}=q} \cdot \Prob[][\s]{Y^{[a,b]}_{T_C}=y, X_{T_C}=q \given Z_k}
    \]
    Let $E\lrc{y,q}$ denote the event $Z_k \,\cap\, Y^{[a,b]}_{T_C}=y\, \cap \,X_{T_C}=q $
    \begin{align*}
        \expectation[][\s]{T^{[a,b]}_a \given E\lrc{y,q}} &= \expectation[][\s]{T_C \given E\lrc{y,q}} + \expectation[][q]{T^{[a-y,b-g]}_{a-g}} \\
        &\geq \expectation[][q]{T^{[a+k \cdot R,b+k \cdot R]}_{a+k \cdot R}} \\
        &\geq \lrc{k+1} \cdot \lrc{\frac{1}{\delta}-1} + \ceil{ \lrc{\log_c(\delta \cdot (1-c))}} \quad \lrc{\eqref{eq:expectation_left_bdry_hitting_time_bscc}}
    \end{align*}
    Summing over all mutually exclusive events, one obtains the specified bound.
\end{proof}

Using \Cref*{lem:lower_bound_probability_hitting_BSCC,lem:lower_bound_expectation_left_bdry_hitting_time_given_Z_k}, one gets the following result.

\begin{restatable}{lemma}{lemboundsgenmc}\label{lem:lower_bound_expectation_left_bdry_hitting_time_gen}
    For any $k > n$, $0 < \delta < 1$, with $a = \min\lrc{-R \ceil{ \lrc{\log_c(\delta \cdot (1-c))}} + R-1, -h} - k \cdot R$ and $b = 1$
    \begin{align*}
        \expectation[][\s]{T^{[a,b]}_a} &\geq \lrc{1-2\cdot g^k} \cdot \lrc{\lrc{k+1} \cdot \lrc{\frac{1}{\delta}-1} + \ceil{ \lrc{\log_c(\delta \cdot (1-c))}}} \\
        \expectation[][\s]{T^{[a,b]}_b} &\leq \size{\states} + \frac{2}{1-g} + \frac{\abs{a}+b+h+R}{\mu}
    \end{align*}
    where $g$ and $h$ are computable constants dependent only on $\mc$ and $c$ depends on $\mc$ along with reward function $r$. Furthermore, if $r_2: E \to \set{-R,\dots,0,\dots,R}$ is an additional reward function with positive mean payoff of at least $\mu_2$ in every BSCC, then assuming $\delta$ is small enough such that $\lrc{\size{\states}+1} \cdot \lrc{\frac{1}{\delta}-1} + \ceil{ \lrc{\log_c(\delta \cdot (1-c))}} \geq \frac{h}{\mu_2}$
    \begin{align*}
        \expectation[][\s]{\lrc{Y_{T^{[a,b]}_a}}_2} &\geq \lrc{\lrc{\lrc{\lrc{k+1} \cdot \lrc{\frac{1}{\delta}-1} + \ceil{ \lrc{\log_c(\delta \cdot (1-c))}}} \cdot \mu_2} \cdot \lrc{1-2g^k} -h} \\
        &-R \cdot \lrc{ \size{\states} + \frac{2}{1-g}} \\
        &-R \cdot \frac{2g}{\lrc{1-g}^2}
    \end{align*}
\end{restatable}
\begin{proof}
The first result follows from \eqref{eq:expectation_left_bdry_hitting_time_lower_bnd} and \Cref{lem:lower_bound_expectation_left_bdry_hitting_time_given_Z_k,lem:lower_bound_probability_hitting_BSCC}.

To get the upper bound for $T^{[a,b]}_b$ in general case, we simply add the expected time it takes to reach a BSCC and upper bound the time it takes with the worst possible $Y^{[a,b]}_{T_C}\lrc{= a}$. Hence, by \Cref{lem:lower_bound_probability_hitting_BSCC,lem:bound_hitting_time_b}
$$\expectation[][\s]{T^{[a,b]}_b} \leq \expectation[][\s]{T_C} + \expectation[][q]{T^{[0,b-a]}_b} \leq \size{\states} + \frac{2}{1-g} + \frac{\abs{a}+b+h+R}{\mu}  $$

Finally, for the lower bound on $\expectation[][\s]{\lrc{Y_{T^{[a,b]}_a}}_2}$, if $\expectation[][\s]{T^{[a,b]}_a} = \infty$ we are done as $\mu_2 > 0$  could be a lower bound in this case for achievable mean payoff making $\expectation[][\s]{\lrc{Y_{T^{[a,b]}_a}}_2} = \infty$, so assume $\expectation[][\s]{T^{[a,b]}_a} < \infty$.

By law of total expectation, partitioning based on whether $T_C < T^{[a,b]}_a$, we have
\begin{align*}
    \expectation[][\s]{\lrc{Y_{T^{[a,b]}_a}}_2} &= \expectation[][\s]{\lrc{Y_{T^{[a,b]}_a}}_2 \given T^{[a,b]}_a < T_C} \cdot \Prob[][\s]{T^{[a,b]}_a < T_C} \\ 
    &+ \expectation[][\s]{\lrc{Y_{T^{[a,b]}_a}}_2 \given T^{[a,b]}_a \geq T_C} \cdot \Prob[][\s]{T^{[a,b]}_a \geq T_C}
\end{align*}
We will now show lower bounds for each summand separately. For $\expectation[][\s]{\lrc{Y_{T^{[a,b]}_a}}_2 \given T^{[a,b]}_a < T_C} \cdot \Prob[][\s]{T^{[a,b]}_a < T_C}$, since $r_2\lrc{e} \geq -R$ always, 
$$\expectation[][\s]{\lrc{Y_{T^{[a,b]}_a}}_2 \given T^{[a,b]}_a < T_C} \geq -R \cdot \expectation[][\s]{T^{[a,b]}_a \given T^{[a,b]}_a < T_C} \label{eq:expected_sum_lower_bound_-R_times_expected_time}$$
\begin{align*}
   & \expectation[][\s]{T^{[a,b]}_a \given T^{[a,b]}_a < T_C} \cdot
     \Prob[][\s]{T^{[a,b]}_a < T_C}\\
  &= \sum_{m=k}^{\infty} m \cdot \Prob[][\s]{T^{[a,b]}_a = m \given T^{[a,b]}_a < T_C} \cdot \Prob[][\s]{T^{[a,b]}_a < T_C} \\
    &= \sum_{m=k}^{\infty} m \cdot \Prob[][\s]{T^{[a,b]}_a = m \cap T^{[a,b]}_a < T_C} \\
    &\leq \sum_{m=k}^{\infty} m \cdot \Prob[][\s]{T_C > m} \\
    &\leq \sum_{m=k}^{\infty} m \cdot 2 \cdot g^{m} \quad \lrc{\Cref{lem:lower_bound_probability_hitting_BSCC}} \\
    &\leq \sum_{m=0}^{\infty} m \cdot 2 \cdot g^{m}  \\
    &= \frac{2g}{\lrc{1-g}^2}
\end{align*}
 and then using above inequality and \eqref{eq:expected_sum_lower_bound_-R_times_expected_time}
$$ \expectation[][\s]{\lrc{Y_{T^{[a,b]}_a}}_2  \given T^{[a,b]}_a < T_C} \cdot \Prob[][\s]{T^{[a,b]}_a < T_C} \geq -R \cdot \frac{2g}{\lrc{1-g}^2} $$

For the other summand, we split the sum into two parts; sum of rewards gained until $T_C$ \ie, until one reaches a BSCC and sum of rewards inside a BSCC. Let $T \eqdef T^{[a,b]}_a - T_C$ denote the time spent inside a BSCC and $\lrc{Y_T}_2$ denote the sum of rewards inside the BSCC before hitting $a$. Then by linearity of expectation
\[
    \expectation[][\s]{\lrc{Y_{T^{[a,b]}_a}}_2 \given T^{[a,b]}_a \geq T_C} \cdot \Prob[][\s]{T^{[a,b]}_a \geq T_C} = \expectation[][\s]{\lrc{Y_{T_C}}_2 + \lrc{Y_T}_2 \given T^{[a,b]}_a \geq T_C} \cdot \Prob[][\s]{T^{[a,b]}_a \geq T_C} \label{eq:expected_sum_split_before_and_after_bscc}
\]
For $\lrc{Y_{T_C}}_2$, one can follow a similar structure to that of the previous summand. So we first find an upper bound on $\expectation[][\s]{T_C \given T^{[a,b]}_a \geq T_C} \cdot \Prob[][\s]{T^{[a,b]}_a \geq T_C}$
\begin{align*}
    \expectation[][\s]{T_C \given T^{[a,b]}_a \geq T_C} \cdot \Prob[][\s]{T^{[a,b]}_a \geq T_C} &\leq \expectation[][\s]{T_C} \quad \textrm{Since } T_C \geq 0 \\
    &= \sum_{m=0}^{\infty} \Prob[][\s]{T_C > m} \\
    &= \size{\states} + \frac{2 \cdot g^{\size{\states}}}{1-g} \quad \lrc{\Cref{lem:lower_bound_probability_hitting_BSCC}} \\
    &\leq \size{\states} + \frac{2}{1-g}
\end{align*}
Thus
$$\expectation[][\s]{\lrc{Y_{T_C}}_2 \given T^{[a,b]}_a \geq T_C} \cdot \Prob[][\s]{T^{[a,b]}_a \geq T_C} \geq -R \cdot \lrc{ \size{\states} + \frac{2}{1-g}}. \label{eq:expected_sum_before_bscc}$$

To calculate expectation for $\lrc{Y_T}_2$, we condition on the energy level $Y^{[a,b]}_{T_C}$ w.r.t $r$ and the state in which we enter the BSCC. 
Suppose $Y^{[a,b]}_{T_C} = y$ and $X_{T_C} = q$. But conditioned on $Y^{[a,b]}_{T_C} = y \cap X_{T_C} = q \cap T^{[a,b]_a} \geq T_C$, $T$ is precisely $T^{[a-y,b-y]}_{a-y}$ with $X_0 = q$ 
\begin{align*}
    \expectation[][\s]{\lrc{Y_T}_2 \given T^{[a,b]}_a \geq T_C \cap X_{T_C} = q \cap \lrc{Y^{[a,b]}_{T_C}}_1 = y} &= \expectation[][q]{\lrc{Y_{T^{[a-y,b-y]}_{a-y}}}_2 \given T\geq 0} \\ 
    &= \expectation[][q]{\lrc{Y_{T^{[a-y,b-y]}_{a-y}}}_2}
\end{align*}
Also assume the mean payoff w.r.t $r_2$ in this BSCC is some $\lambda > \mu_2 > 0$.
Since $T$ is a stopping time with finite expectation (as per our assumption), we can apply optional stopping theorem to the martingale of $r_2$ cf. ~\Cref{fact:martingale}.% R#3: l.1001
\[
  {m^q_2}_n = \lrc{Y_n}_2 + \nu\lrc{X_n} - n\lambda  
\]
 where the index $n$ is actually counted from $T_C$. This implies
\begin{align*}
{m^q_2}_T &= {m^q_2}_0 \\
\lrc{Y_T}_2 + \nu\lrc{X_T} - T \lambda &= \nu\lrc{q} \\
\lrc{Y_T}_2 &= T\, \lambda + \nu\lrc{q} - \nu\lrc{X_T} \\
&\geq T\, \mu_2 - h \\
\implies \expectation[][q]{\lrc{Y_T}_2} &\geq \expectation[][q]{T} \, \mu_2 - h
&=  \expectation[][q]{T^{[a-y,b-y]}_{a-y}} \cdot \mu_2 - h
\end{align*}
Therefore, partitioning over all possible tuples $\lrc{q,y}$
\begin{align*}
& \expectation[][\s]{\lrc{Y_T}_2 \given T^{[a,b]}_a \geq T_C} \cdot
                 \Prob[][\s]{T^{[a,b]}_a \geq T_C} \\
  & \geq \sum_{(q,y)}\lrc{\expectation[][q]{T^{[a-y,b-y]}_{a-y}} \cdot
                 \mu_2 -h} \cdot \Prob[][\s]{T^{[a,b]}_a \geq T_C \cap X_{T_C}
                 = q \cap \lrc{Y^{[a,b]}_{T_C}}_1 = y}
\end{align*}                 
When $Y^{[a,b]}_{T_C} \geq -k \cdot R$, then $\expectation[][q]{T^{[a-y,b-y]}_{a-y}} \geq \lrc{k+1} \cdot \lrc{\frac{1}{\delta}-1} + \ceil{ \lrc{\log_c(\delta \cdot (1-c))}}$ using \Cref{eq:expectation_left_bdry_hitting_time_bscc}. 
Observe that this event subsumes $T_C < k$, so probability of this happening is $\geq 1 - 2g^k$ by ~\Cref{lem:lower_bound_probability_hitting_BSCC}. 
In the other case when ${Y_{T_C}}_1 < -k \cdot R$, a trivial lower bound of $0$ for $\expectation[][q]{T^{[a-y,b-y]}_{a-y}}$ suffices for our purposes. 
Since from our assumption, $\delta$ is small enough so that $\lrc{\size{\states}+1} \cdot \lrc{\frac{1}{\delta}-1} + \ceil{ \lrc{\log_c(\delta \cdot (1-c))}} \geq \frac{h}{\mu_2}$, putting it all together we have
\begin{align*}
& \expectation[][\s]{\lrc{Y_T}_2 \given T^{[a,b]}_a \geq T_C} \cdot
  \Prob[][\s]{T^{[a,b]}_a \geq T_C} \\
  & \geq \lrc{\lrc{k+1} \cdot
  \lrc{\frac{1}{\delta}-1} + \ceil{ \lrc{\log_c(\delta \cdot (1-c))}} \cdot
  \mu_2-h} \cdot \lrc{1-2g^k} - h \cdot 2g^k
\end{align*}  

So
\begin{align*}
    \expectation[][\s]{\lrc{Y_{T^{[a,b]}_a}}_2} &\geq -R \cdot \frac{2g}{\lrc{1-g}^2} \\
    &+ \lrc{-R \cdot \lrc{ \size{\states} + \frac{2}{1-g}}} \\
    &+ \lrc{\lrc{\lrc{\lrc{k+1} \cdot \lrc{\frac{1}{\delta}-1} + \ceil{ \lrc{\log_c(\delta \cdot (1-c))}}} \cdot \mu_2} \cdot \lrc{1-2g^k} -h}.
\end{align*}
\end{proof}

\section{Proof of \texorpdfstring{\Cref{lem:lower_bounds_expected_sums_gain_bailout}}{Lemma 10}}\label{app:bounds}

\constantsvonevtwo*
\begin{proof}
    We parametrise $\size{\mdp^{\ast}}$ along with $\vec{r}$ on 
\begin{itemize}
    \item Number of states $n \eqdef \size{\states^\ast}$.
    \item Maximum bit length of probability in $\probp$. Let it be $w$.
    \item Number of reward dimensions $d$.
    \item Maximum reward on an edge in any dimension $R$.
\end{itemize}
One can also similarly define size of the Markov chain induced by some finite-memory strategy.

Let $f\lrc{n,w,d,R} \eqdef n^2\lrc{2+w+d\cdot \lrc{1 + \ceil{\lrc{\log_2\lrc{R+1}}}}}$.
Assuming binary representation of rewards, it is easy to see that $\size{\mdp} \leq f\lrc{n,w,d,R}$. The probabilities are always represented in binary.

As $\optzstrat_{\Bailout}$ is MD, $\size{\mc^{\Bailout}} \leq f\lrc{n,w,d,R}$. \\ 
Similarly, as $\optzstrat_{\Gain}$ is obtained as a result of a linear program, we have 
$$\size{\mc^{\Gain}} \leq f\lrc{n,LP_\Gain\lrc{n,w,d,R},d,R}$$ 
where $LP_\Gain$ (cf. \Cref{fig:LP_gain}) is some fixed polynomial in $n,w,d$ and $\log R$. For succinctness, we define $w_{\Gain} \eqdef LP_\Gain\lrc{n,w,d,R}$.

\begin{figure}
    \centering
    \begin{align*}
        \max \eps \\
        \sum_{\s \in C} \pi_\s &= 1 & \pi_\s : \textrm{Avg. time spent in } \s \\
        \sum_{\tuple{\s,\s'} \in \transition_C} x_{\tuple{\s,\s'}} &= \pi_\s & \s \in \zstates \cap C \\
        x_{\tuple{\s,\s'}} &= \pi_\s \cdot \probp\lrc{\s}\lrc{\s'} & \s \in \rstates \cap C \\
        \sum_{\tuple{\s,\s'} \in \transition_C} x_{\tuple{\s,\s'}} \cdot \vec{r}\lrc{\tuple{\s,\s'}} &\geq \eps \cdot R & \vec{\MP}_{[1,d]}\lrc{> \vec{0}} \\
        \eps &\geq 0 & \\
    \end{align*}
    \caption{LP for an MEC $C$ for $\Gain$ \cite[Figure 3]{brazdil2014markov}}
    \label{fig:LP_gain}
\end{figure}

To show that all the constants lie in $\Ocompl\lrc{\exp\lrc{\size{\mdp^\ast}^{\Ocompl\lrc{1}}}}$, we simply consider $v_1^i + v_2^i$ and show that each of the constants for $v_1^i + v_2^i$ is in the required size. The result then follows by observing that every constant is positive.

From \Cref{lem:sufficiency_condition_pos_mp,lem:lower_bound_expectation_left_bdry_hitting_time_gen}, for 
$$Z_g \eqdef Z_b + R + k \cdot R + \max\lrc{R \ceil{ \lrc{\log_{c_i}(\delta \cdot (1-c_i))}} + R-1, h_{\Gain}}$$
it suffices to choose $k$ and $\delta$ such that
\begin{align}
    \lrc{\lrc{k+1} \cdot \lrc{\frac{1}{\delta}-1} + \ceil{ \lrc{\log_{c_i}(\delta \cdot (1-c_i))}}} \cdot \mu_i \cdot \lrc{1-2 g_\Gain^k}&-  \nonumber\\
    h_\Gain - R \cdot \lrc{ \lrc{n + \frac{2}{1-g_\Gain}} - \frac{2g_\Gain}{\lrc{1-g_\Gain}^2} }&> \nonumber\\
    R \cdot \lrc{n + \frac{2}{1-g_{\Bailout}} + \frac{Z+h_{\Bailout}+R}{\mu}}& \quad \forall\; 2 \leq i \leq d \label{eq:v_1^i > -v_2^_i}\\
    k &> n \\
    \lrc{\lrc{k+1} \cdot \lrc{\frac{1}{\delta}-1} + \ceil{ \lrc{\log_{c_i}(\delta \cdot (1-c_i))}}} \cdot \mu_i &\geq h_\Gain \quad \forall\; 2 \leq i \leq d
\end{align}

The last set of equations become redundant due to the first one. The left-hand side of \eqref{eq:v_1^i > -v_2^_i} is the over precision of constant $v_1^i$ and similarly the right-hand side is that of $-v^2_i$. It is simple to notice that once $k$ is fixed, then $v_1^i$ varies with $\delta$ as $ C_1 \cdot  \frac{1}{\delta} - C_2 \log\lrc{\frac{1}{\delta}} - C_3$ and $v_2^i$ as $ - C_4 \log\lrc{\frac{1}{\delta}} - C_5$ for some appropriate constants $C_i$. To further simplify, W.l.o.g, assume $\delta$ is sufficiently small such that 
$$R \ceil{ \lrc{\log_{c_i}(\delta \cdot (1-c_i))}} + R \geq h_{\Gain}.$$  
By definition of $Z_g$ and our assumption on $\delta$, we get that 
$$Z_g = Z_b + R + k \cdot R + R \ceil{ \lrc{\log_{c_i}(\delta \cdot (1-c_i))}} + R-1.$$
Then rearranging constants to one side and terms depending on $k$ and $\delta$ on other side we get
\begin{align*}
    \lrc{\lrc{k+1} \cdot \lrc{\frac{1}{\delta}-1} + \ceil{ \lrc{\log_{c_i}(\delta \cdot (1-c_i))}}} \cdot \mu_i \cdot \lrc{1-2 g_\Gain^k}& \\ 
    -\lrc{k + \ceil{ \lrc{\log_{c_i}(\delta \cdot (1-c_i))}}} \cdot \frac{R^2}{\mu} &> \\
    h_\Gain + R \cdot \lrc{ \lrc{n + \frac{2}{1-g_\Gain}} + \frac{2 \cdot g_\Gain}{\lrc{1-g_\Gain}^2} }& \\
    + R \cdot \lrc{n + \frac{2}{1-g_{\Bailout}} + \frac{Z_b + 3 \cdot R - 1+h_{\Bailout}}{\mu}}&
\end{align*}
We will upper bound the RHS and lower bound LHS by simpler formulas to get sufficient bounds on $k$ and $\delta$. 
Let $x_{\min,1}$ denote the minimum probability in $\mc^{\Gain}$, and $x_{\min,2}$ denote the minimum probability in $\mc^{\Bailout}$. 
Then by definition of $w_\Gain$ and $w$,
$$ x_{\min,1} \geq \frac{1}{2^{w_\Gain}} \quad x_{\min,2} \geq \frac{1}{2^w} $$
from \eqref{def:h}
\begin{align*}
    h_\Gain &= \frac{2\cdot n \cdot R}{x_{\min,1}^n} \\
    &\leq 2^{1 + \ceil{\log_2 n + \log_2 R} + n \cdot w_\Gain} \\
    &\leq 2^{f\lrc{n,w_\Gain,d,R}}
\end{align*}
Similarly, one gets $h_{\Bailout}\leq  2^{f\lrc{n,w,d,R}}$.

From ~\eqref{def:g}
\begin{align}
    1 - g_\Gain &=  1 - \exp{\lrc{\frac{-x_{\min,1}^{n}}{n}}} \nonumber \\
    &\geq \frac{x_{\min,1}^n}{2n} \lrc{\textrm{ Since } 1-e^{-x} \geq \frac{e-1}{e} \cdot x \geq \frac{x}{2} \textrm{ for } x \in [0,1]} \label{eq:ineq_g} \\
    \implies  R \cdot \frac{2}{1-g_\Gain} &\leq \frac{4\cdot n \cdot R}{x_{\min,1}^n} \\
    &= 2h_\Gain \\
    &\leq 2\,2^{f\lrc{n,w_\Gain,d,R}}
\end{align}
 Similarly, $R \cdot \frac{2}{1-g_\Bailout} \leq 2\,2^{f\lrc{n,w,d,R}} $
\begin{align*}
    R \cdot \frac{2g_\Gain}{\lrc{1-g_\Gain}^2} &\leq R \cdot \frac{2}{\lrc{1-g_\Gain}^2}  \\
    &\leq \frac{2 \cdot n^2 \cdot R}{x_{\min,1}^{2n}} \\
    &\leq 2^{1 + \ceil{2\log n + \log R} + 2n \cdot w_\Gain} \\
    &\leq 2^{2f\lrc{n,w_\Gain,d,R}}
\end{align*}

To get a lower bound on $\mu$, let $\?{B}$ be any BSCC of $\mc^{\Bailout}$ and $\probp_{\?{B}}$ be the one-step transition probability matrix in $\mc^{\Bailout}$ restricted to $\?{B}$. Clearly the number of states in $\?{B}$ is $\leq n$. The steady state probabilities $\pi_{\?{B}}$ are solution to the linear system \Cref{fig:LP_steady_state}.
\begin{figure} 
    \centering
    \begin{align*}
        \lrc{I - \probp_{\?{B}}}^T \cdot \pi_{\?{B}} &= \vec{0} \\
        \sum_{\s \in \?{B}} \pi_{\?{B}}\lrc{\s} &= 1 \\
        \pi_{\?{B}} &\geq \vec{0}
    \end{align*}
    \caption{$\textrm{LP}_1$: Linear program for steady state probabilities in a BSCC}
    \label{fig:LP_steady_state}
\end{figure}

We apply ~\cite[Theorem 15]{cmuLP}. First, lets multiply each row by lcm of denominators to get integer entries. The size of each entry is now bounded by $n \cdot w$. Therefore, size of the entire matrix is $\leq n^3 \cdot w$. $size\lrc{b}$ here $\leq 2n+2$. $\implies$ the denominator of each component of $\pi_{\?{B}}$ is $\leq 2^{\lrc{n^3w+2n+2}} \leq 2^{f^2\lrc{n,w,d,R}}$.

The mean payoff in this BSCC is then given by 
$$ \mu_{\?{B}} \eqdef \sum_{\s} \sum_{\setcomp{\s'}{\lrc{\s,\s'} \in \transition_{\?{B}}}} \pi_{\?{B}}\lrc{\s} \cdot \probp\lrc{\s}\lrc{\s'} \cdot r_1\lrc{\lrc{\s,\s'}}$$

The least common denominator for all such $\probp\lrc{\s}\lrc{\s'}$ will be $\leq 2^{n\cdot w} \leq 2^{f\lrc{n,w,d,R}}$ which means the overall denominator for $\mu_{\?{B}} \leq 2^{f^2\lrc{n,w,d,R} + f\lrc{n,w,d,R}} \leq 2^{2f^2\lrc{n,w,d,R}}$.

Therefore, $\mu_{\?{B}} \geq 2^{-2f^2\lrc{n,w,d,R}}$. Since $\mu$ is just minimum over all such $\mu_{\?{B}}$,
$$\mu \geq 2^{-2f^2\lrc{n,w,d,R}}$$

Finally, $Z_b = \max_\s i^{\Bailout}_\s \leq 3 \cdot n \cdot R$. Combining everything and from the fact that $w_\Gain \geq w$, we get that RHS 
\begin{align*}    
    &\leq 2^{f\lrc{n,w_\Gain,d,R}} + 2\cdot2^{f\lrc{n,w_\Gain,d,R}} + 2^{2f\lrc{n,w_\Gain,d,R}} + 2 \cdot n \cdot R \\
    &+ 2\cdot2^{f\lrc{n,w,d,R}} + R^2 \cdot 2^{2f^2\lrc{n,w,d,R}} \lrc{3 \cdot n \cdot R + 3\cdot R + 2^{3f^2\lrc{n,w,d,R}}} \\
    &\leq 5 \cdot 2^{f\lrc{n,w_\Gain,d,R}} + 2 \cdot n \cdot R + 2^{2f\lrc{n,w_\Gain,d,R}} \\
    &+ R^2 \cdot 2^{2f^2\lrc{n,w,d,R}} \lrc{3 \cdot \lrc{n+1} \cdot R + 2^{3f^2\lrc{n,w,d,R}}} \\
    &\leq 7 \cdot 2^{2f\lrc{n,w_\Gain,d,R}} + 2^{2f\lrc{n,w,d,R}} \cdot 2^{2f^2\lrc{n,w,d,R}} \lrc{2^{2f\lrc{n,w,d,R}} + 2^{3f^2\lrc{n,w,d,R}}} \\
    &\leq 2^{2f\lrc{n,w_\Gain,d,R} + 3} + 2^{4f^2\lrc{n,w,d,R}} \cdot 2^{5f^2\lrc{n,w,d,R}} \\
    &\leq 2 \cdot 2^{9f^2\lrc{n,w_\Gain,d,R}}
\end{align*}
To lower bound LHS, first choose $k$ to be sufficiently large such that $g^k_\Gain \leq 1/4$. Let $k = \ceil{\frac{2n}{x_{\min,1}^n}} \geq n + 1$. Then
\begin{align*}
    \lrc{\lrc{k+1} \cdot \lrc{\frac{1}{\delta}-1} + \ceil{ \lrc{\log_{c_i}(\delta \cdot (1-c_i))}}} \cdot \mu_i \cdot \lrc{1-2 g_\Gain^k} &\geq \lrc{k+1} \cdot \lrc{\frac{1}{\delta}-1} \cdot \frac{\mu_i}{2} \\
    &\geq \frac{k \mu_i}{2 \delta} \quad \textrm{Assume } \delta < \frac{1}{k+1} \\
    &\geq 2^{-2f^2\lrc{n,w_\Gain,d,R}} \cdot \frac{1}{\delta} \\ 
    &\textrm{since } k \geq 2 \textrm{ and } \mu_i \geq 2^{-2f^2\lrc{n,w_\Gain,d,R}}
\end{align*}

\begin{align*}
    \frac{k\,R^2}{\mu} &\leq k\,R^2\; 2^{f\lrc{n,w,d,R}} \\
    &\leq 2^{2\,f\lrc{n,w_\Gain,d,R}} \cdot 2^{f\lrc{n,w,d,R}} \\
    \implies -\frac{k\,R^2}{\mu} &\geq 2^{3\,f\lrc{n,w_\Gain,d,R}}
\end{align*}
From ~\eqref{def:c} and ~\eqref{def:eta}
\begin{align*}
    \log_{c_i}(\delta \cdot (1-c_i)) \cdot \frac{R^2}{\mu} &= \frac{\log 1/\delta + \log (1/\lrc{1-c_i})}{\log 1/c_i} \cdot \frac{R^2}{\mu} \\
    &\leq \lrc{\log 1/\delta + \log (1/\lrc{1-c_i})} \cdot \frac{2 \eta_i^2 \cdot R^2}{\mu_i^2 \cdot \mu}
\end{align*}
Using $\eta_i \leq 3 \cdot h_\Gain$, and $1 - c_i \geq \frac{\mu_i^2}{2\eta_i^2}$, we get
\begin{align*}
    &\leq \lrc{\log_2 \lrc{1/\delta} + 1 + 2\cdot \log_2\eta_i + 2\cdot\log_2 1/\mu_i } \cdot  \lrc{3h_\Gain R}^2 \cdot 2^{6f^2\lrc{n,w_\Gain,d,R}} \\
    &\leq \lrc{\log_2 \lrc{1/\delta} + 1 + 2\cdot \log_2 3 + 2f\lrc{n,w_\Gain,d,R} + 4f^2\lrc{n,w_\Gain,d,R}} \cdot 9\cdot 2^{10f^2\lrc{m,w_\Gain,d,R}} \\
    &\leq \lrc{\log_2 \lrc{1/\delta} + 5 + 6\cdot f^2\lrc{n,w_\Gain,d,R}} \cdot 9\cdot 2^{10f^2\lrc{n,w_\Gain,d,R}} \\
    &\leq  9 \cdot 2^{10f^2\lrc{n,w_\Gain,d,R}} \log_2 \lrc{\frac{1}{\delta}} + 99\cdot 2^{11f^2\lrc{n,w_\Gain,d,R}}
\end{align*}

Combining everything, we have LHS
\begin{equation}
    \geq 2^{-2f^2\lrc{n,w_\Gain,d,R}} \cdot \frac{1}{\delta} - 9 \cdot 2^{10f^2\lrc{n,w_\Gain,d,R}} \log_2 \lrc{\frac{1}{\delta}} - 99 2^{11f^2\lrc{n,w_\Gain,d,R}} - 2^{3f\lrc{n,w_\Gain,d,R}}
\end{equation}

Comparing it with the constants from \Cref{lem:lower_bounds_expected_sums_gain_bailout}, one can see that $C_1 =  2^{-2f^2\lrc{n,w_\Gain,d,R}}$, $C_2 + C_3 = 9 \cdot 2^{10f^2\lrc{n,w_\Gain,d,R}}$ and $C_4 + C_5 = 102\, 2^{11f^2\lrc{n,w_\Gain,d,R}}$ all of which are in the required complexity.

Finally, it suffices to choose a $\delta$ such that
\begin{align*}
    2^{-2f^2\lrc{n,w_\Gain,d,R}} \cdot \frac{1}{\delta} - 9 \cdot 2^{10f^2\lrc{n,w_\Gain,d,R}} \log_2 \lrc{\frac{1}{\delta}} &>  102\, 2^{11f^2\lrc{n,w_\Gain,d,R}}
\end{align*}

$\delta = 2^{-20f^2\lrc{n,w_\Gain,d,R}}$ should satisfy the required inequality. Therefore, with $k = \ceil{\frac{2n}{x_{\min,1}^n}}$, and $\delta = 2^{-20f^2\lrc{n,w_\Gain,d,R}}$ the overall bound $b$ will be exponential in $\size{\mdp}$.
\end{proof}

\section{Proofs of \texorpdfstring{\Cref{sec:mem_bounds}}{memory bounds section}}\label{app:mem_bounds}

Before proceeding with the proof of \Cref{lem:BSCC_mem_bound}, we state some useful definitions and prove some intermediate lemmas 
which makes it easier to understand the idea. We start by defining the notion of a winning end component (WEC).
\begin{definition}\label{def:WEC}
    Let $\mdp\lrc{B} = \tuple{\states_B,{\zstates}_B, {\rstates}_B,
      \transition_B, \vec{r}_B}$ be an end component of $\mdp$. 
    We say that $\mdp\lrc{B}$ is a WEC (winning end component) iff there is some strategy $\zstrat \in \zfinstrats{\mdp\lrc{B}}$ such that 
    \begin{itemize}
        \item ${\mdp\lrc{B}}^{\zstrat}$ is irreducible and the end component defined by it is exactly $\mdp\lrc{B}$.
        \item for every state $q\,\in\,\states_B$, there is some minimal energy level $j_q$ such that $\Prob[\mdp\lrc{B}][\zstrat,q]{\obj\lrc{j_q}} = 1$.
    \end{itemize}
\end{definition}

We simply say $B$ is a WEC instead of $\mdp\lrc{B}$ is a WEC for succinctness.

Furthermore, denote by $\mathbb{C}\lrc{B} \eqdef \setcomp{\mathsf{C}}{\mathsf{C} \textrm{ is a simple cycle in } \mdp\lrc{B}}$, the set of all simple 
cycles in $\mdp\lrc{B}$ and given a simple cycle $\mathsf{C} = \s_0 \energymove{c_0} \s_1 \energymove{c_1} \ldots \s_j = \s_0$ be a cycle of length 
$j$, where $c_i$ denotes the rewards in the energy $\lrc{1^{st}}$ dimension, define the effect of $\mathsf{C}$ to be $\mathtt{eff}\lrc{\mathsf{C}} \eqdef \sum_{k=0}^{j-1} c_k$.

A WEC $B$ is called a WEC of Type-I if there is some $\mathsf{C}\, \in \, \mathbb{C}\lrc{B}$ such that $\mathtt{eff}\lrc{\mathsf{C}} > 0$. 
Otherwise, it is called a WEC of Type-II.

\begin{lemma}\label{lem:WEC_Type-I => MP_1 > 0}
    If $B$ is a WEC of Type-I, then one can choose $\zstrat$ such that it satisfies all the conditions of \Cref{def:WEC} along with
    $$ \Prob[\mdp\lrc{B}][\zstrat,q]{\MP_1\lrc{> 0}} = 1 $$ for every state $q \in \states_B$.
\end{lemma}
\begin{proof}
    Assume that $\zstrat = \memstrattuple$ which satisfies the requirements of \Cref{def:WEC} gives a mean payoff of $0$ in the energy dimension. Let $\mathfrak{C} = \tuple{\s_0,\memconf_0} \energymove{c_0} \tuple{\s_1,\memconf_1} \energymove{c_1} \ldots \tuple{\s_k,\memconf_k} = \tuple{\s_0,\memconf_k}$ be a simple cycle of length $k$ in ${\mdp\lrc{B}}^\zstrat$.
    \begin{claim}\label{claim:Composite_cycle_effect_T_2}
        For any cycle $\mathfrak{C}$ in $\mdp^{\zstrat}$, $\mathtt{eff}\lrc{\mathfrak{C}} = 0$.
    \end{claim}
    \begin{claimproof}
        We have $\Prob[\mdp\lrc{B}][\zstrat,\s]{\obj\lrc{j_\s}} = 1$, and $\Prob[\mdp][\zstrat,\s]{\MP_1 \lrc{> 0}} = 0$. The former implies that $\MP_1\lrc{\geq 0}$ surely. In fact, it can be never be the case that $\mathtt{eff}\lrc{\mathfrak{C}} < 0$ as otherwise $\en_1$ and hence $\obj\lrc{j_\s}$ is not satisfied almost surely. If $\mathtt{eff}\lrc{\mathfrak{C}} > 0$ for some $\mathfrak{C}$, this then implies a positive mean payoff since $\mdp^{\zstrat}$ is an irreducible, finite Markov chain, a contradiction. Hence, $\mathtt{eff}\lrc{\mathfrak{C}} = 0$.
    \end{claimproof}
    We construct a strategy $\zstrat'$ which follows \Cref{def:WEC} such that it
    almost surely satisfies positive meanpayoff in the energy dimension along with $\obj\lrc{j_\s}$ \ie
    $\Prob[\mdp\lrc{B}][\zstrat',\s]{\obj\lrc{j_\s}\,\cap\,\MP_1\lrc{> 0}} = 1$. Since $B$ is a WEC of Type-I, there is some cycle $\mathsf{C} = \s_0 \energymove{c_0} \ldots \energymove{c_{\ell-1}} \s_{\ell} = \s_0$ with positive effect. And since effect of every cycle in $\mdp^{\zstrat}$ is $0$, this implies that the reward along \emph{any} path between two given states of $\mdp^{\zstrat}$ must be identical. Also, every edge in $\transition_B$ occurs somewhere in $\mdp^\zstrat$ by definition of $\zstrat$. Let the edge $\s_i \energymove{c_i} \s_{i+1}$ in $\mathsf{C}$ occur at memory mode $\memconfa_i$ \ie $\tuple{\s_i,\memconfa_i} \energymove{c_i} \tuple{\s_{i+1},\memconfa^{\prime}_i}$. $\memconfa^{\prime}_i$ may or may not be the same as $\memconfa_{i+1}$. However, irreducibility of the Markov chain implies there are paths $p_{i+1}$ connecting $\tuple{\s_{i+1},\memconfa^{\prime}_i}$ to $\tuple{\s_{i+1},\memconfa_{i+1}}$. Consider the cycle $\tuple{\s_0,\memconfa_0} \energymove{c_0} \tuple{\s_1, \memconfa^{\prime}_0} \energymove{\mathtt{eff}\lrc{p_1}} \tuple{\s_1,\memconfa_1} \ldots \energymove{\mathtt{eff}\lrc{p_\ell}} \tuple{\s_0,\memconfa_0}$. The effect of the entire cycle is $0$ by~\cref{claim:Composite_cycle_effect_T_2} and $\mathsf{C}$ is part of this cycle. This implies that the sum of effects of all the paths $p_{i+1}$ is negative and therefore at least one $p_{i+1}$ has negative effect. The strategy $\zstrat'$ simply bypasses this path and updates the memory mode directly to $\memconfa_{i+1}$ with some small probability.
    For some sufficiently small $\eps > 0$
    \begin{itemize}
        \item with probability $\lrc{1-\eps}$ follow $\zstrat$ at $\tuple{\s_i, \memconfa_i}$
        \item with probability $\eps$, directly move to $\tuple{\s_{i+1},\memconfa_{i+1}}$
    \end{itemize}
    Observe that it doesn't matter if $\s_i$ was a random or a controlled state as the final destination for both edges is the 
    same with only the memory mode being different, so $\zstrat'$ is a valid strategy which updates its memory stochastically.

    It is easy to see that $\zstrat'$ also induces an irreducible Markov chain with every edge in $\transition_B$ occurring at least once, and that the 
    energy objective $\en\lrc{j_\s}$ is satisfied as the shortcut introduced has a positive effect on the energy 
    level. Furthermore, for sufficiently small $\eps$, it doesn't change the mean payoff in other dimensions by much thereby still ensuring that 
    $\zstrat'$ satisfies $\vec{\MP}_{[2,d]}\lrc{> \vec{0}}$. Finally, the addition of this new edge now causes the mean payoff in $1^{st}$ dimension to be strictly $ > 0$ as there is now at least one (complex) cycle with positive weight and still no cycles with negative weight from the properties of $\zstrat$.
\end{proof}

\Cref{lem:WEC_Type-I => MP_1 > 0} shows that it is possible to win both $\Gain$ and $\Bailout$ almost surely in
$\mdp\lrc{B}$ from every state in $q \in B$ whenever $B$ is a Type-I WEC.
I.e., $q \in \AS^{\mdp\lrc{B}}\lrc{\vec{\MP}_{[1,d]} \lrc{> \vec{0}}}$.
Moreover, the minimal safe energy for $\Bailout$ in $\mdp\lrc{Q}$ from $q$ is
exactly $j_q$, that is
$q \in \AS^{\mdp\lrc{B}}\lrc{\en_1\lrc{j_q}\, \cap \, \MP_1\lrc{> 0}}$.
Thus, $\mdp\lrc{B}$ satisfies the conclusion of
\Cref{lem:gain_bailout_existence}, i.e., it behaves like $\mdp^\ast$.

Therefore, the strategy $\optzstrat_{\mathtt{alt,Z_b,Z_g,b}}$, defined in \Cref{sec:inf_str ==> fin_str},
is almost surely winning $\obj\lrc{j_q}$ in $\mdp\lrc{B}$.
We can now carry over the analysis on the memory bound $b$ for $\mdp^{\ast}$ from
\Cref{lem:lower_bounds_expected_sums_gain_bailout,lem:sufficiency_condition_pos_mp}
to $\mdp\lrc{B}$.
The only difference is that the size is now measured in
$|\mdp\lrc{B}| \le |\mdp|$. So we obtain the following lemma.

\begin{lemma}\label{lem:bound_ell_1}
    If $B$ is a WEC of Type-I, there exists a bound $b_B = \Ocompl\lrc{\exp\lrc{\size{\mdp\lrc{Q}}^{\Ocompl\lrc{1}}}}$ such that for all $q \in B$
    \[
        \Prob[\mdp\lrc{B}][\optzstrat_{\mathtt{alt,Z_b,Z_g,b_B}},q]{\obj\lrc{j_q}} = 1.
    \]
\end{lemma}

By \Cref{rem:infix}, it is also true that $\Prob[\mdp\lrc{B}][\optzstrat_{\mathtt{alt,Z_b,Z_g,b_B}},q]{\obj\lrc{j_q} \, \cap \, \infix\lrc{b_B}} = 1$.

Note that, if there is no positive effect cycle in $B$, there cannot be any cycle with negative effect as well since every cycle is taken infinitely 
often in a WEC. 
So, in contrast to Type-I, if $B$ is such that $\mathtt{eff}\lrc{\mathsf{C}} = 0$ for every simple cycle $\mathsf{C}$, 
then $B$ is called a WEC of Type-II. But this implies that the maximum fluctuation in energy level is at most $\size{\states_B} \cdot R \leq \size{\states} \cdot R$. Therefore, we get the following.

\begin{lemma}\label{lem:bound_ell_2}
    For every WEC $B$ of Type-II, there is a finite-memory strategy
    $\zstrat_B$
    with a constant $b_B\, \in \, \Ocompl\lrc{\size{\states} \cdot R}$ such that
    $$\Prob[\mdp\lrc{B}][\zstrat_B,\s]{\obj\lrc{j_\s} \, \cap \, \infix\lrc{b_B}} = 1.$$
\end{lemma}

We are now ready to prove \Cref{lem:BSCC_mem_bound}.
\bsccmemboundTthree*
\begin{proof}
We provide the bounds based on the type of the end component $\mdp\lrc{B}$. First observe that $B$ is a WEC as $\zstrat$ acts as a witness by satisfying 
the requirements of \Cref{def:WEC}. If $B$ is a WEC of Type-II, then by \Cref{lem:bound_ell_2} the minimal energy $j_q$ required to win from any state 
$q$ in $\states_{B}$ is $\leq \size{\states_B} \cdot R \leq \size{\states} \cdot R$ and the constant $b_B$ is bounded by 
$\Ocompl\lrc{\size{\states} \cdot R}$. Choose $\zstrat_B$ and $b_B$ be the strategy and the constant from \Cref{lem:bound_ell_2} in this case.
Otherwise, $B$ is a WEC of Type-I. Therefore, $j_q$ in this case would be the same as the minimal energy to satisfy $\Bailout$ which
by \Cref{lem:md_bailout_strategy} is $\leq 3  \size{\states_{B}}  R \leq  3  \size{\states}  R$. By choosing $\zstrat_B$ as 
$\optzstrat_{\mathtt{alt,Z_b,Z_g,b_B}}$, with $b_B$ from \Cref{lem:bound_ell_1}, we are done.
\end{proof}

\transientmemboundTthree*
\begin{proof}
It is clear that $\zstrat$ is also a witness for $\en_1\lrc{i_\s}\, \cap\, \eventually\,T$. 
But by \cite[Lemma 2]{CD2011}, this can be achieved with at most $2  \size{\states}  R$ fluctuation in energy. 
However, since we also need to ensure that the necessary minimal energy level, one can then simply encode the energy level into the state space
of $\mdp$ and enlarge $\mdp$ up to $\lrc{3+2}  \size{\states}  R$. So the states of this new MDP $\mdp^{\prime}$ will now be 
$\states \x [0,5  \size{\states}  R]$. 
Let $T^{\prime} = \bigcup_{ q \in \states_{B}} q \x [i^{B}_q, 5 \size{\states}  R]$. Then, it is not hard to see that when starting from $\tuple{\s,i_\s}$ in $\mdp^{\prime}$, one almost surely satisfies $\eventually\, T^{\prime}$. (Move according to $\zstrat$ until you hit the maximum energy level of 
$5 \size{\states} R$, at which point switch to one of the winning strategies for 
$\en_1\lrc{\cdot}\, \cap\, \eventually\, T$ which uses only $\Ocompl\lrc{\size{\states}\cdot R}$ memory modes.)
Therefore, one can reach a state $q$ in a BSCC with its safe energy level with a fluctuation of at most $5\cdot\size{\states}\cdot R$.
\end{proof}

\end{document}